\documentclass[aps,prx,reprint,superscriptaddress,nofootinbib,longbibliography]{revtex4-2}

\usepackage{amsmath, amssymb, amsthm, mathtools}
\usepackage{dsfont}
\usepackage{physics}
\usepackage{quantikz}
\newtheorem{theorem}{Theorem}
\newtheorem{claim}[theorem]{Claim}
\newtheorem{proposition}[theorem]{Proposition}

\usepackage[ruled,vlined,linesnumbered,noend]{algorithm2e}
\usepackage{graphicx}
\usepackage{subcaption}
\makeatletter
\AtBeginDocument{%
	\long\def\@makecaption#1#2{%
		\par
		\vskip\abovecaptionskip
		\begingroup
			\small\rmfamily
			\samepage
			\flushing
			\let\footnote\@footnotemark@gobble
			\@make@capt@title{#1}{#2}\par
		\endgroup
		\vskip\belowcaptionskip
	}%
}
\makeatother

\SetKwProg{Fn}{def}{:}{}

\usepackage[colorlinks=true, citecolor=blue, linkcolor=blue, urlcolor=blue]{hyperref}

\begin{document}
	
\title{Quantum algorithms for four problems in density peak clustering}

\author{Gonçalo Sousa}
\affiliation{Centro de Física das Universidades do Minho e do Porto, LaPMET, Departamento de Física e Astronomia, Faculdade de Ciência, Universidade do Porto, Rua do Campo Alegre, Porto 4169-007, Portugal}

\author{Duarte Magano}
\affiliation{Centro de Física das Universidades do Minho e do Porto, LaPMET, Departamento de Física e Astronomia, Faculdade de Ciência, Universidade do Porto, Rua do Campo Alegre, Porto 4169-007, Portugal}
\affiliation{21strategies GmbH, Hallbergmoos, Germany}

\date{\today}

\begin{abstract} 
Clustering is a fundamental task in unsupervised learning.
The goal is to partition unlabeled data into clusters such that similar elements are assigned to the same cluster, while elements in different clusters are dissimilar.
We propose a quantum Monte Carlo approach to density estimation, a key routine in density peak clustering.
Building on this approach, we develop quantum algorithms for the full clustering problem and three related tasks: decision clustering, cluster counting, and heavy cluster counting.
For each one, we establish a polynomial quantum query speedup under assumptions on the scaling of dataset-dependent parameters. 
Numerical experiments suggest that these conditions are satisfied for several practically relevant settings.
\end{abstract}

\maketitle

\section{Introduction \label{sec:intro}}

Machine learning algorithms infer patterns from data to make predictions or decisions with minimal human intervention~\cite{bishop2006pattern}. Besides their well-known commercial use~\cite{netflix}, such algorithms are being increasingly used in scientific research~\cite{jumper_highly_2021,carrasquilla_machine_2017,raissi2019}.
As the size and complexity of datasets increase, however, the training and deployment of learning models become more computationally demanding~\cite{icml2008,large_scale_ML}.
Demonstrated quantum speedups for specific computational problems have motivated research into quantum algorithms for machine learning~\cite{Biamonte_2017}.
For example, \citet{wiebe_quantum_2012} employed a quantum primitive for solving linear systems~\cite{Harrow_2009} to design a data-fitting algorithm, while \citet{lloyd_quantum_2014} developed an efficient approach to principal component analysis based on quantum simulation~\cite{LLoyd1996}.

Most of the works in quantum machine learning fit the supervised learning paradigm, where the model learns to classify new data based on already labeled data. 
In contrast, in unsupervised learning we look for some structure in a dataset in the absence of previously established labels.
Here we focus on clustering, a canonical example of unsupervised learning.
It can be stated as follows.

\begin{table*}[!t]
\begin{ruledtabular}
\begin{tabular}{lcccc}
& Clustering & Decision clustering & Cluster counting & Heavy cluster counting \\[4pt]
\hline
\noalign{\vskip 4pt}
Classical DPC
& $N^2$
& $N^2$
& $N^2$
& $N^2$ \\[12pt]

Quantum DPC
& $\dfrac{N}{\epsilon} + N^{3/2}$
& $\dfrac{N}{\epsilon} + \sqrt{N}H$
& $\dfrac{N}{\epsilon} + \sqrt{NK(N-K+1)}$
& $\dfrac{N}{\epsilon} + \sqrt{N}H\left(\dfrac{K}{m}+\dfrac{1}{\eta^{2}}\right)$\\
\noalign{\vskip 4pt}
\end{tabular}
\end{ruledtabular}
\caption{\label{tab:2} Query complexity of four clustering-related problems: the clustering, the decision clustering, cluster counting and heavy cluster counting. Each column corresponds to one of these problems. The complexities refer to classical DPC (first row), and to quantum DPC (second row). The variable $N$ is the size of the dataset; $K$ is the number of clusters; $\epsilon$ is the error introduced in density estimates; $H$ is the maximum length of any path formed by successively moving to nearest neighbors of greater density up to a cluster center; $m$ is the mass threshold in heavy cluster counting, and $\eta$ the respective tolerance. All logarithmic dependencies were excluded for simplicity.}
\end{table*}

\vspace{0.2cm}
\noindent\textbf{Clustering.} Consider a dataset $S$ and a distance function $d: S \times S \to \mathbb{R}$ (acting as a dissimilarity measure). Partition $S$ into sets, called \textit{clusters}, such that similar elements belong to the
same cluster and dissimilar elements belong to distinct clusters.
\vspace{0.2cm}

Evidently, the notion of cluster as we have presented it is mathematically ill-defined. 
There are many models for defining clusters and, for each, many ways to construct such clusters, resulting in a large variety of clustering algorithms~\cite{many_clusterings}. Centroid-based clustering algorithms, such as k-means~\cite{lloyd_kmeans} and k-medoids~\cite{kmedoids}, assume the existence of a centroid for each cluster, and search for the partitions that minimize some criterion function favoring intra-cluster similarity. Hierarchical clustering algorithms, like single-linkage~\cite{sneath_application_1957}, construct a dendrogram by successively merging clusters according to their similarity.
Another class of clustering algorithms is known as density-based clustering, where clusters are viewed as contiguous regions of high density of points~\cite{hartigan_clustering_1975}. 
Density-based clustering algorithms typically have some attractive features not found in algorithms for other models, such as not making assumptions about cluster shapes and not requiring the number of clusters to be given as input~\cite{campello_density-based_2020}.
	
Clustering, like other machine learning problems, has also received attention from the quantum computing community. The first step was taken by~\citet{aimeur2007}, who used quantum search routines to speed up essential subtasks of divisive clustering and k-medoids~\cite{kmedoids}. As another example,~\citet{li2011hybrid} modeled data elements as quantum walkers on a lattice. \citet{lloyd_quantum_2013} quantized a routine for finding the nearest centroid, which is a building block of k-means~\cite{lloyd_kmeans}. Building on this,~\citet{wiebe_quantum_2014} developed a coherent algorithm for nearest-neighbor classification, which can also be applied to k-means. Later, a new quantum version of k-means was developed~\cite{qmeans}. More recently, spectral clustering has also been brought into the quantum setting~\cite{kerenidis_quantum_2021}.

Density peak clustering (DPC)~\cite{Laio2014} is a popular density-based algorithm. 
DPC calculates a `density' for every element of the dataset, and then assigns each element to the same cluster as the nearest neighbor that has a density greater than itself.
Denoting $N$ as the size of the dataset, this algorithm scales as $\mathcal{O}(N^2)$, making it challenging to apply to large datasets.
To address this issue,~\textcite{Magano_2023} introduced quantum algorithm techniques for DPC.
Failing to provide a speedup for the full clustering problem, they studied a decision version.

\vspace{0.2cm}
\noindent\textbf{Decision clustering.} Consider a dataset $S$, a distance function $d: S \times S \to \mathbb{R}$, and two elements $x_i, x_j \in S$. Decide whether $x_i$ and $x_j$ belong to the same cluster.
\vspace{0.2cm}

\noindent The decision version of DPC has the same $\mathcal{O}(N^2)$ complexity as the full problem.
Nevertheless,~\textcite{Magano_2023} showed that this problem can be solved with a quantum algorithm with complexity $\widetilde{\mathcal{O}}(H N^{3/2})$, where $H$ is the maximum length of any path formed by successively moving to nearest neighbors of greater density up to a cluster center.
They argued that, under certain assumptions, $H$ scales as $\mathcal{O}(N^{1/d_{\mathrm{eff}}})$, where $d_{\mathrm{eff}}$ is related to the intrinsic dimension of the dataset.
Under these assumptions, their algorithm provides a quantum advantage when $d_{\mathrm{eff}} > 2$.

In this work, we extend these results towards a quantum speedup for the full clustering problem.
Our key observation is that each density can be expressed as an expectation of a random variable and can thereby be estimated using Monte Carlo sampling.
Although classical Monte Carlo density estimation does not improve upon the complexity of classical DPC, this formulation enables the use of quantum Monte Carlo methods based on amplitude estimation~\cite{montanaro_quantum_2015, Brassard_2002}.
Combining this approach with the techniques developed by~\textcite{Magano_2023}, we obtain a quantum DPC algorithm scaling as $\widetilde{\mathcal{O}}(N \epsilon^{-1} + N^{3/2})$, where $\epsilon$ is the additive error in the density estimates.
We should pick this precision to be sufficiently small to recover the cluster structure that would have been identified by the original DPC.
For a given dataset, we define $\epsilon_\mathrm{crit}$ as the maximum error for which the original clusters are recovered.
The efficiency of our algorithm depends on how $\epsilon_\mathrm{crit}$ scales with $N$, finding quantum speedup whenever $\epsilon_\mathrm{crit} = \widetilde{\omega}(N^{-1})$.
We provide heuristic arguments and numerical evidence suggesting that, in several practically relevant settings, $\epsilon_\mathrm{crit}$ depends weakly on the dataset size $N$, remaining approximately constant over the range of $N$ considered.
However, we also construct a family of datasets for which $\epsilon_\mathrm{crit} = \Theta(N^{-1})$, demonstrating that our approach does not provide an unconditional quantum speedup.

We revisit the decision clustering problem in the light of quantum Monte Carlo density estimation, obtaining a quantum algorithm with complexity $\widetilde{\mathcal{O}}(N\epsilon^{-1}+\sqrt{N}H)$.
This improves upon the algorithm of~\textcite{Magano_2023} whenever $H\epsilon_\mathrm{crit}=\widetilde{\omega}(N^{-1/2})$, a condition that, again, we expect to hold in many practically relevant settings. The decision clustering algorithm beats the full clustering algorithm (for the decision task) if $H = o(N)$.

We also broaden the study of quantum DPC by introducing the problem of counting the number of clusters in a dataset, regardless of their content and structure.

\vspace{0.2cm}
\noindent\textbf{Cluster counting.} Consider a dataset $S$ and a distance function $d: S \times S \to \mathbb{R}$. Determine the number $K$ of clusters in $S$.
\vspace{0.2cm}

\noindent To the best of our knowledge, there has been little prior work treating cluster counting as a standalone task. Usually, the number of clusters is assumed as input or it is inferred only as a by-product of cluster reconstruction.

Classical DPC requires $\mathcal{O}(N^2)$ time to count the number of clusters.
We make use of the quantum exact counting algorithm~\cite{Brassard_2002} to provide a solution to this problem with complexity $\widetilde{\mathcal{O}}(N\epsilon^{-1} +  \sqrt{N K (N-K+1)})$.
For a constant number of clusters, this result constitutes a speedup by a factor of $\sqrt{N}$ over the quantum clustering.

Instead of counting the number of clusters exactly, we can try to estimate it, aiming for a more efficient solution. Inspired by the approaches to the unseen-species problems~\cite{fisher1943,efron_estimating_1976}, we develop a sample-based algorithm to produce such an estimate. In the context of DPC, we can sample a given cluster $C_k$ with probability equal to its mass, $m_k =|C_k|/N$. 
A dataset may contain many low-mass clusters (in the extreme, single-element clusters), such that we would require an enormous amount of samples to estimate their number accurately. Moreover, in some applications, it may be sufficient to count only those clusters whose mass exceeds a prescribed threshold.
We then propose the following variant of cluster counting, where we are only concerned with \emph{heavy} clusters.

\vspace{0.2cm}
\noindent\textbf{Heavy cluster counting} 
Consider a dataset $S$ and a distance function $d: S \times S \to \mathbb{R}$. 
Define the $m$-heavy cluster count $K_m$ as
\begin{equation}
	K_m = \vert \{k = 1,\dots, K : m_k > m \} \vert.
\end{equation}
Let $m > \eta > 0$.
With constant probability, output an estimate $\widehat{K_m}$ of the $m$-heavy cluster count satisfying
\begin{equation}
	K_{m + \eta} \leq \widehat{K_m} \leq K_{m - \eta}.
\end{equation}
\vspace{0.2cm}

We develop a quantum algorithm for this counting variant with complexity $\widetilde{\mathcal{O}}(N\epsilon^{-1} + \sqrt{N}H (Km^{-1} + \eta^{-2}))$.
For constant $K$, $m$, and $\eta$ this algorithm is more efficient than the exact quantum cluster counting if $H = o(\sqrt{N})$.

We summarize our results in Table~\ref{tab:2}.
Our algorithms work in the black-box model, in which the dataset is accessed through an oracle that returns the distances between elements.
Accordingly, the stated complexities are query complexities, measured by the number of calls to this distance oracle.
Unlike the original DPC algorithm, the quantum algorithms considered here are probabilistic: they return the correct answer with constant success probability.
This probability can always be amplified to $1-\gamma$ at a cost of $\mathcal{O}(\log(1/\gamma))$ repetitions.

This paper is organized as follows.
In Section~\ref{sec:blackbox} we provide more details about the black-box model.
Section~\ref{sec:dpc} then presents the (classical) DPC algorithm, assuming that the reader has no prior knowledge of the topic. 
Section \ref{sec:mc_dpc} continues by rephrasing density estimation in terms of Monte Carlo sampling.
In Section~\ref{sec:pre} we summarize the background material in quantum algorithms required to understand our work (the reader familiar with these topics may readily skip this part).
Section~\ref{sec:qalgo} presents the key quantum primitives that will be used in our algorithms.
Section~\ref{sec:app} derives the stated results by applying these primitives to the problems of clustering, decision clustering, and cluster counting. 
In Section~\ref{sec:num} we analyze the dependence of the critical error on structural features of the dataset, including its size.
Section~\ref{sec:conclusion} concludes with a summary and discussion. 

\section{Computational model \label{sec:blackbox}}

We adopt the black-box model introduced by~\textcite{aimeur2007}, where the knowledge about the distance between two points is solely available through queries to an oracle. Given two labels $i,j$ as inputs to the oracle, this returns the distance between the elements $x_i,x_j \in S$,
\begin{equation}
	(i,j) \xrightarrow{\text{query}} d(x_i,x_j).
\end{equation}
We assume that the distance function constitutes a metric, which is a requirement for the density-based model of clusters~\cite{everitt_cluster_2011}. The complexity measure is the number of distance queries, i.e., the number of calls to the oracle. 

This model is a suitable abstraction when the database is composed exclusively of distances between data elements, or when there is an external routine that estimates the distances between elements without apparent structure.
Nevertheless, the black-box model may not be a good description of some common settings, as it conceals any structure underlying the dataset.
For example, if we know the coordinates of a set of points in Euclidean space we can resort to computational geometry methods, such as k-d trees~\cite{kdtrees}, R*-trees~\cite{r_star_trees}, or Voronoi diagrams~\cite{voronoi_diagrams}, that may speed up clustering tasks.

In the quantum setting, the oracle is represented by a unitary operator, $O_d$, acting as
\begin{equation}
	O_d\ket{b}\ket{i}\ket{j} = \ket{b \oplus d(x_i,x_j)}\ket{i}\ket{j},
\end{equation}
where $b$ is a bitstring of length $n_d$, the number of bits needed to store distances up to the desired accuracy. The quantum registers carrying the element labels are composed of $n = \lceil \log N \rceil$ qubits. One can act with this oracle in a given superposition of states with different pairs of labels $i,j$, and different bitstrings $b_{ij}$,
\begin{equation}
	O_d \left(\sum_{i,j} \alpha_{ij}\ket{b_{ij}}\ket{i}\ket{j}\right) = \sum_{i,j} \alpha_{ij} \ket{b_{ij} \oplus d(x_i,x_j)}\ket{i}\ket{j},
\end{equation}
for any set of normalized complex amplitudes $\{\alpha_{ij}\}$.
The quantum query complexity is the number of applications of the operator $O_d$.

\section{Density peak clustering}
\label{sec:dpc}

\begin{figure}
	\begin{subfigure}{0.49\columnwidth}
	\includegraphics[width=\linewidth]{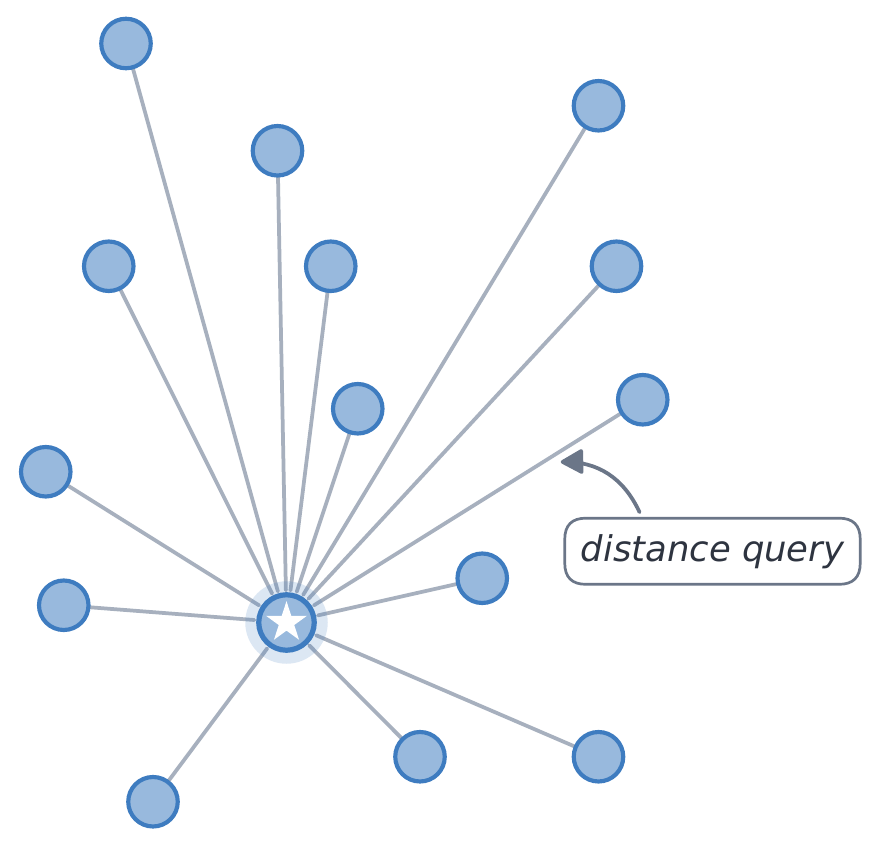}
	\caption{}
	\end{subfigure}
	\begin{subfigure}{0.49\columnwidth}
	\includegraphics[width=\linewidth]{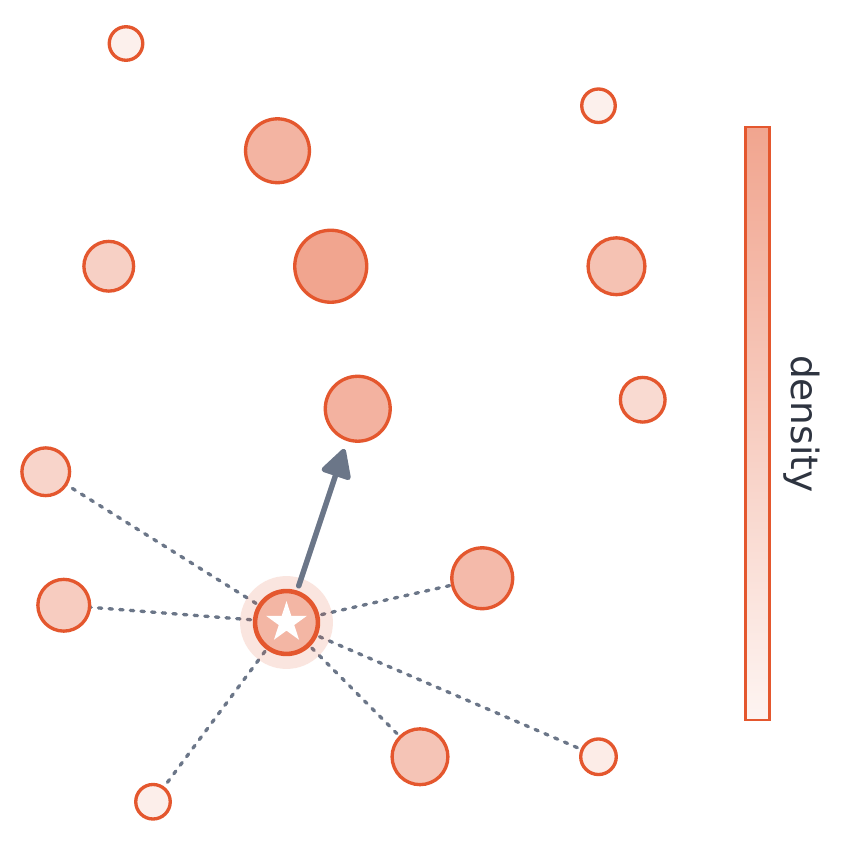}
	\caption{}
	\end{subfigure}
\caption{\textbf{Schematic representation of two key steps in DPC}. \textbf{(a)} The density of the starred element is computed by querying its distance to all other elements in the dataset. \textbf{(b)} The nearest-higher of the starred element is the nearest neighbor with a density greater than itself; elements are colored and sized according to their density; dotted lines link the starred element to some neighbors with smaller density; the solid arrow links the starred element to its nearest-higher.}
\label{fig:DPC_steps}
\end{figure}

The density peak clustering (DPC) algorithm was proposed in 2014 by~\textcite{Laio2014}, and since then it has been widely studied and applied~\cite{WANG2024121860,WEI2023126633,dpc_theoretical}.
Similar to other density-based algorithms, it does not require specifying the number of clusters in advance and can identify arbitrarily shaped clusters.
In addition, it can discriminate outliers, making it useful for machine learning problems other than clustering, such as anomaly detection~\cite{tu_hyperspectral_2020}. 

The DPC algorithm assigns a local \emph{density} value, $\rho_i$, to every element $x_i \in S$, defined as follows
\begin{equation}
	\label{eq:density}
	\rho_i \coloneqq \frac{1}{N}\sum_{x_j \in S} \chi\big(d(x_i,x_j) / d_c\big),
\end{equation}
where $\chi(\cdot)$ is a kernel function of the distance between elements, and $d_c$ is a scale parameter. 
We will also assume that $\chi \leq 1$, so that $\rho_i \leq 1$ for every $x_i$.  Common  choices for the kernel function are the Heaviside kernel, $\chi(x) = \Theta(1- x)$, and the Gaussian kernel, $\chi(x) = \exp(-x^2)$. 
Building on the density values comes the notion of \textit{nearest-higher}.
We say that the nearest-higher of an element $x_i$, denoted as $h_i$, is the nearest neighbor of $x_i$ that has a density greater than itself,
\begin{equation}
	\label{eq:nearest_higher}
	h_i \coloneqq \mathop{\arg\min}_{x_j : \rho_j > \rho_i} d(x_i, x_j).
\end{equation}
The nearest-higher of the point with maximum density is conventionally taken to be the element furthest from it. 
We call the distance between an element and its nearest-higher the \textit{nearest-higher separation}, denoted as
\begin{equation}
	\label{eq:nearest_higher_sep}
	\delta_i = d(x_i,h_i).
\end{equation}
DPC rests on the evaluation of these quantities -- $\rho_i$, $h_i$, and $\delta_i$ -- for all elements of the dataset.
We illustrate this process in Fig. \ref{fig:DPC_steps}.
	
Clustering proceeds based on the following observation: for
elements that are local maxima of the density the nearest-higher separation is much larger than the typical nearest-neighbor distance. 
These local maxima are called \textit{peaks} or \textit{roots}, and they will be at the origin of the cluster formation.
The exceptions are elements with high nearest-higher separation but low density, which are classified as \textit{outliers}.
More specifically, for a given choice of threshold parameters $(\delta_c,\rho_c)$, we classify $x_i \in S$ as a root if 
\begin{equation}
	\label{root}
	\delta_i \geq \delta_c \quad \text{and} \quad \rho_i \geq \rho_c,
\end{equation}
and as an outlier if 
\begin{equation}
	\delta_i \geq \delta_c \quad \text{and} \quad \rho_i < \rho_c.
\end{equation}
All other elements are assigned to the same cluster as their respective nearest-higher. The algorithm is summarized in Appendix~\ref{app:clustering}, and an application example is illustrated in Fig.~\ref{fig:DPC}.

\begin{figure*}
\begin{subfigure}[b]{\columnwidth}
	\centering
	\includegraphics[width=0.8\linewidth]{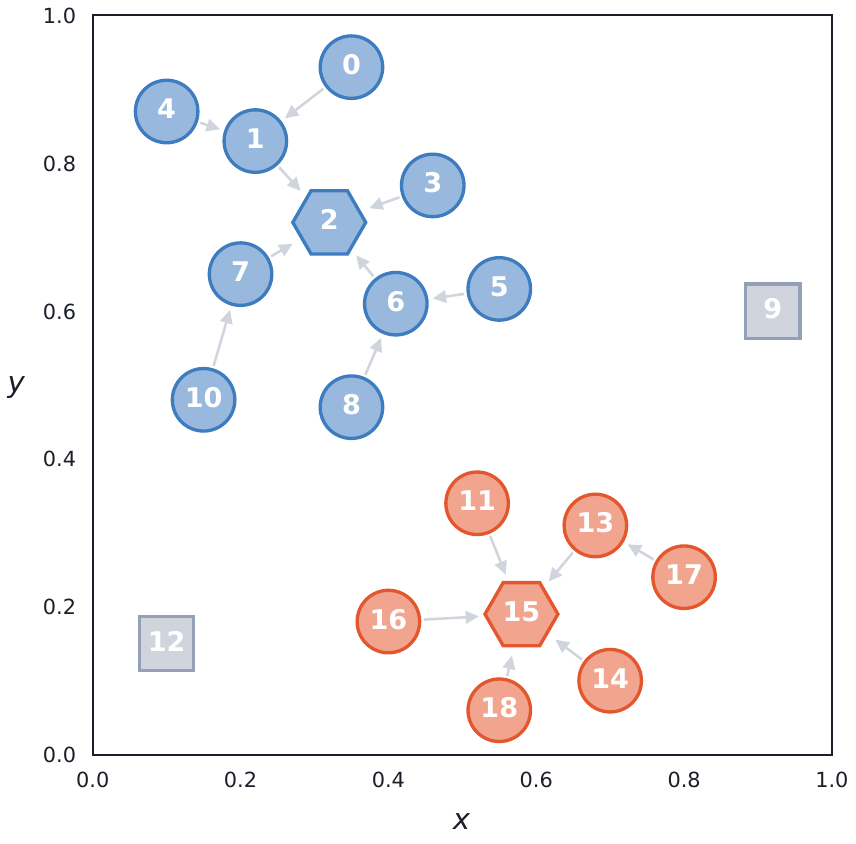}
	\caption{Clusters and trees}
	\label{fig:DPC1}
\end{subfigure}
\hfill
\begin{subfigure}[b]{\columnwidth}

	\centering
	\includegraphics[width=0.8\linewidth]{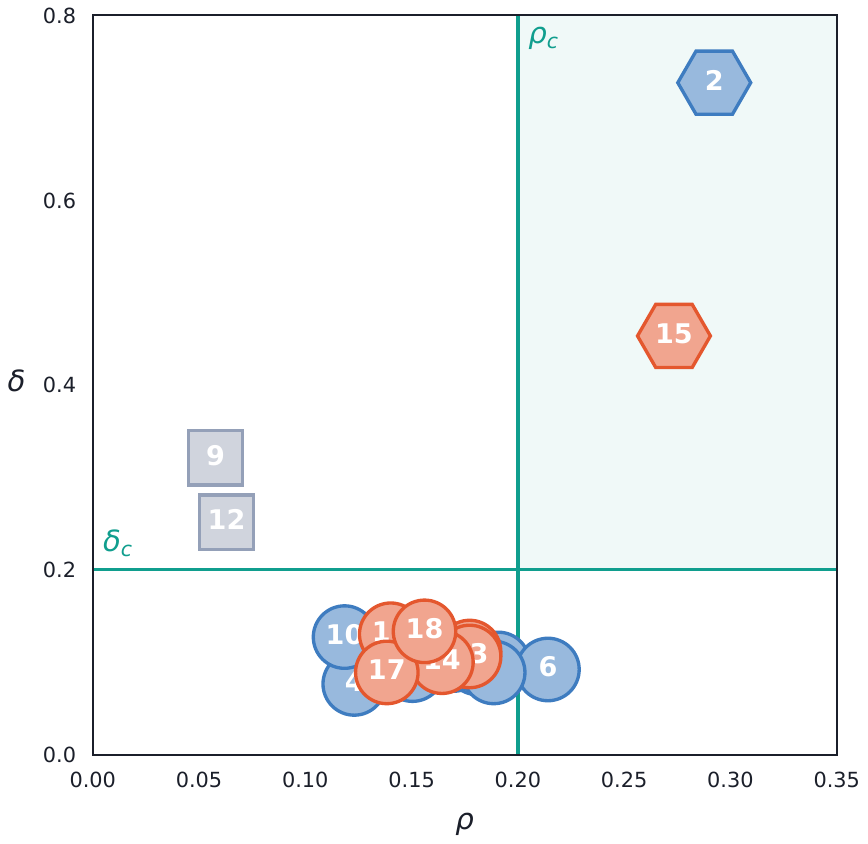}
	\caption{Decision graph}
	\label{fig:DPC2}
\end{subfigure}
\caption{\label{fig:DPC} \textbf{Density peak clustering.} \textbf{(a)} Clustering dataset with 19 points embedded in $\mathbb{R}^2$, using the Euclidean distance. Roots are represented by hexagons, outliers by squares, and the remaining points by circles, having arrows pointing to their nearest-highers. Each color (blue, red) identifies a cluster. \textbf{(b)} Decision graph displaying the nearest-higher separation $\delta$ against density $\rho$. Roots (elements $2$ and $15$) satisfy $\delta_i \geq \delta_c$ and $\rho_i \geq \rho_c$, while outliers (elements $9$ and $12$) satisfy $\delta_i \geq \delta_c$ and $\rho_i < \rho_c$.}
\end{figure*}

The scaling of DPC is dominated by the computation of the nearest-highers, which requires knowing the densities of all elements in the dataset.
To compute the density of a given element we need to query the distance to every other element.
Repeating this for all elements in the dataset means that we end up querying all $(N^2 - N)/2 = \mathcal{O}(N^2)$ distances.
This scaling, as opposed to the subquadratic scaling of some alternative non-black-box methods such as DBSCAN~\cite{dbscan} and OPTICS~\cite{optics}, has been identified as one of the main drawbacks of DPC~\cite{WEI2023126633}.
	
\section{Monte Carlo density estimation}
\label{sec:mc_dpc}

We propose a Monte Carlo approach to estimating the density values.
It is based on the observation that densities can be viewed as expected values of random variables representing the kernel function and can be estimated with a Monte Carlo method.
While it does not imply an improved scaling over the original formulation of DPC, it will serve as the basis for the quantum clustering algorithm that we will present later, which does provide a (conditional) speedup over the classical DPC.
Moreover, it will open the way to speedups in clustering subtasks, namely decision clustering, cluster counting, and heavy cluster counting.

Let $J$ be the random variable representing the index of an element sampled from the dataset, and $X_i(J)$ the kernel function evaluated on the distance between that element and a fixed element $x_i \in S$. The expected value of $X_i(J)$ is given by,
\begin{equation}
	\mathbb{E}[X_i(J)] = \sum_{x_j \in S} \mathrm{P}[J=j] \ \chi({d(x_i,x_j)/d_c}),
\end{equation}
where $\mathrm{P}[J=j]$ denotes the probability of sampling the element $x_j \in S$. 
For a uniform probability distribution over the dataset elements, the expected value of $X_i(J)$ is numerically equal to the density of $x_i$. 
Hence, we can take the sample mean estimator, denoted here as $\tilde{\rho}_i$, as an unbiased estimator for the density.
For a sample $\{x_{j_1},\dots,x_{j_M}\}$ of $M$ elements picked uniformly at random from the dataset,
\begin{equation}
	\tilde{\rho}_i = \frac{1}{M}\sum_{m=1}^M \chi({d(x_i,x_{j_m})/d_c}).
\end{equation}
It is worth noting that the original and Monte Carlo approaches are strictly incomparable.
While the former works with the exact density values, computed directly from Eq.~(\ref{eq:density}), the latter produces estimates of these values up to a given error. 

As the kernel function is upper bounded by $1$, the variance of the estimator is upper bounded by $1/4M$.
Following Hoeffding's inequality~\cite{hoeffding1963}, we can estimate $\rho_i$ up to error $\epsilon$, with probability of at least $1-\gamma$, using 
\begin{equation}
	M = \mathcal{O}\left(\frac{1}{\epsilon^2} \log \left(\frac{1}{\gamma}\right)\right)
\end{equation}
random samples and the same number of distance queries.
To perform clustering, we repeat this sampling process for all the densities in the dataset.
But this does not gather information about the nearest neighbor structure, so we still need to perform more distance queries to find the nearest-highers and complete the cluster formation.
The overall scaling is thus 
\begin{equation}
	\mathcal{O}\left( \frac{N}{\epsilon^2} \log (N) + N^2\right),
\end{equation}
providing no advantage over the original version.

The cluster structure identifiable by DPC depends on the density hierarchy in the dataset, which is established once densities have been assigned to every data point. For each dataset, there is a threshold value for the estimation error below which the hierarchy of the Monte Carlo densities induces the same cluster structure as the original DPC, with high probability.
We refer to this threshold as \emph{critical error}, and we denote it by $\epsilon_\mathrm{crit}$.
It  becomes crucial to understand how $\epsilon_\mathrm{crit}$ depends on the size and structure of the datasets. 
In Section~\ref{sec:num} we provide numerical evidence suggesting that, for several practically relevant settings, $\epsilon_\mathrm{crit}$ is approximately constant with respect to the dataset size $N$.
But we also show that it is possible to construct a family of datasets for which $\epsilon_\mathrm{crit} = \Theta(N^{-1})$, demonstrating that $\epsilon_\mathrm{crit}$ depends in a complex way on the dataset structure.

\section{Preliminaries on quantum algorithms}
\label{sec:pre}

\subsection{Quantum minimum finding}
\label{sec:qmf}

The quantum minimum finding algorithm is the key routine in the clustering algorithm of~\textcite{Magano_2023}. 
It was developed by~\textcite{durr1999}, and it is based on Grover's search algorithm~\cite{grover}.
Here we provide a brief review of these techniques.
	
Consider a domain $\mathcal{T} = \{1,\dots,N\}$ and a boolean function $f: \mathcal{T}  \to \{0,1\}$. 
The function works as a black box, in the sense that the information about it can only come from evaluating it over its domain. 
The goal is to find an element $i \in \mathcal{T} $ such that $f(i)=1$, assuming that there are $t$ such elements.
Classically, we need $\mathcal{O}(N/t)$ queries to $f$ to find such an element with high probability.
Now we introduce a quantum oracle for $f$, $O_f$, that marks the $1$-inputs with a $-1$ phase,
\begin{equation}
	\label{eq:grover}
	O_f \ket{i} = (-1)^{f(i)}\ket{i}.
\end{equation}
That is, $f$ is encoded in the phases of the eigenvalues of $O_f$. Building on the work of~\textcite{grover}, \textcite{Boyer1998} showed that we can solve this search problem with $\mathcal{O}(\sqrt{N/t})$ queries to $O_f$.
The algorithm begins by creating an equal superposition of the $N$ elements in $\mathcal{T}$,
\begin{equation}
	\ket{\psi} = \frac{1}{\sqrt{N}} \sum_{i =1}^N \ket{i}.
\end{equation}
 We can rewrite the initial state as
\begin{equation}
	\label{eq:psi_angle}
	\ket{\psi} = \sin \theta \ket{\psi_1} + \cos \theta \ket{\psi_0},
\end{equation}
where we have defined the angle $\theta \coloneqq \arcsin \sqrt{t/N}$, the superposition of $1$-inputs $\ket{\psi_1} \coloneqq \frac{1}{\sqrt{t}} \sum_{f(i) = 1} \ket{i}$, and the superposition of $0$-inputs $\ket{\psi_0} \coloneqq \frac{1}{\sqrt{N-t}} \sum_{f(i) = 0} \ket{i}$.
If we measure the state $\ket{\psi}$, we obtain any input with equal probability. But we can amplify the probability of measuring a $1$-input state by applying successive rotations. Consider the reflection about the state $\ket{\psi}$
\begin{equation}
	R_\psi = \left(2\dyad{\psi}-I\right),
\end{equation}
and the rotation operator given by $G_f \coloneqq R_\psi O_f$. After $k$ rotations the amplitude associated with $\ket{\psi_1}$ becomes $\sin\left((2k+1)\theta\right)$. It grows roughly linearly with the number of applications of $G_f$, and so the probability grows quadratically, explaining the speedup. The authors provided a further generalization for the case where $t$ is unknown, establishing the following result.
\begin{theorem}[Quantum search,~\cite{Boyer1998}]
	Consider $\mathcal{T}  = \{1,\dots,N\}$, let $f:\mathcal{T}  \to \{0,1\}$, and $t = |\{i \in \mathcal{T} : f(i) = 1 \}|$ (possibly unknown). Given access to the oracle $O_f$, it is possible to find one element $i \in \mathcal{T}$ satisfying $f(i)=1$ using an expected number of $\mathcal{O}(\sqrt{N/t})$ quantum queries to $f$.
\end{theorem}
	
Minimum finding can be formulated as a sequence of search problems. 
Consider now a black-box function $f : \mathcal{T}  \to \mathbb{R}$. 
We want to find $i \in \mathcal{T}$ that minimizes $f$. 
Starting with a random $k \in \mathcal{T}$, we search for an element $j \in \mathcal{T}$ such that $f(j) < f(k)$. 
For that, we construct the function $F_k(j): \mathcal{T} \to \{0,1\}$ as
\begin{equation}
	\label{eq:qmf}
	F_k(j) =
	\begin{cases}
		0, \text{ if } f(j) \geq f(k)\\
		1, \text{ if } f(j) < f(k)
	\end{cases},
\end{equation}
requiring one extra query to $f$.  
We can encode this function in the eigenphases of a unitary operator, such as in Eq.~(\ref{eq:grover}). 
From there, we employ the quantum search to find an element $j$ such that $F_k(j) = 1$, or equivalently, $f(j) < f(k)$. After finding such a $j$, replace $k$ by $j$ and repeat the process. 
\textcite{durr1999} showed that within $\mathcal{O}(\sqrt{N})$ repetitions we find a minimum of $f$ with high probability.
In more detail,
\begin{theorem}[Quantum minimum finding,~\cite{durr1999}]
	\label{thm:qmf}
	Consider $\mathcal{T} = \{1,\dots,N\}$, and let $f:\mathcal{T} \to \mathbb{R}$. Given $\gamma \in ]0,1]$, it is possible to find an element $i \in \mathcal{T}$ that minimizes $f$ with probability at least $1-\gamma$ using $\mathcal{O}(\sqrt{N}\log(1/\gamma))$ quantum queries to $f$.
\end{theorem}
	
\subsection{Quantum amplitude estimation}
\label{sec:qae}

Our clustering algorithm relies on another well-known quantum primitive, the so-called quantum amplitude estimation (QAE). 
Although its original formulation makes essential use of the quantum Fourier transform~\cite{Brassard_2002}, more recent versions rely exclusively on Grover's algorithm~\cite{grover}, such as the one we review here, developed by~\textcite{AR_2020_qae}. 

Let $\ket{\psi}$ be the $n$-qubit state prepared by a given quantum routine $U$. Say that,
\begin{equation}
	U\ket{0^n} = \ket{\psi}= a \ket{\psi_1} + \sqrt{1-a^2} \ket{\psi_0}
\end{equation}
where $a \in [0,1]$, the normalized state $\ket{\psi_1}$ spans the $1$-subspace, orthogonal to the $0$-subspace spanned by 
$\ket{\psi_0}$. 
We are also given access to a quantum oracle that distinguishes  $\ket{\psi_1}$ and  $\ket{\psi_0}$ as
\begin{equation}
	O \ket{\psi_i} = (-1)^i \ket{\psi_i}, \quad i=0,1.
\end{equation}
The goal is to estimate the amplitude $a$. 

Define $\theta \coloneqq \arcsin a$, such that $\ket{\psi}$ is written as in~Eq.~\eqref{eq:psi_angle}.
Let $R_\psi$ be the reflection about $\ket{\psi}$, and $G \coloneqq R_\psi O$ be the Grover rotation operator. Recall that after $k$ applications of $G$ on $\ket{\psi}$, the probability of measuring a state in the $1$-subspace is $\sin^2\left((2k+1)\theta\right)$. 
The rate with which this probability changes with $k$ depends on the value of $\theta$; \textcite{AR_2020_qae} have used this fact to narrow down the range of possible values of $\theta$ iteratively, in a manner akin to binary search.

Suppose that, at some point, $\theta$ is known to be in $[\theta_\mathrm{min},\theta_\mathrm{max}]$. Let $k$ be chosen such that 
\begin{equation}
	\label{eq:conditions_n_rotations}
	(2k+1)\theta_\mathrm{min} \sim q \pi \text{ and } (2k+1)\theta_\mathrm{max} \sim q \pi + \tfrac{\pi}{2},
\end{equation}
for a positive integer $q$. After $k$ Grover rotations, if $\theta \approx \theta_\mathrm{max}$ there is a constant probability of measuring the resulting state in the $1$-subspace; but if $\theta \approx \theta_\mathrm{min}$ instead, that probability becomes vanishingly small. Thus, one uses the estimate for that probability to prune the range of possible values of $\theta$ accordingly. 

The range is shrunk by a constant factor, until its length is no greater than $2\epsilon$. The final midpoint provides an estimate $\tilde{\theta}$ for $\theta$ within error $\epsilon$, and the corresponding estimate for the amplitude $\tilde{a} \coloneqq \sin \tilde{\theta}$ also has an error of $\epsilon$. From conditions in Eq.~\eqref{eq:conditions_n_rotations}, it follows $k \sim(\theta_\mathrm{max}-\theta_\mathrm{min})^{-1}$ in each iteration. Since the range shrinks by a constant factor, the number of rotations employed per iteration increases exponentially. The cost of the last iteration dominates, so that the resulting complexity scales as~$\widetilde{\mathcal{O}}(\epsilon^{-1})$.

The original algorithm in Ref.~\cite{AR_2020_qae} produces an estimate with relative error rather than with additive error, as we have described. A simple adaptation is required to convert the relative error model to the additive one, but the intuition about the algorithm, as well as the structure of the proofs of correctness and complexity remain the same. We provide a description of the adapted algorithm in Appendix~\ref{app:additive_error}. The original result adapted for additive error is stated as follows.

\begin{theorem}[Quantum amplitude estimation, adapted from Ref.~\cite{AR_2020_qae}]
	\label{thm:qae}
	Let $U$ be a quantum routine preparing the state $$\ket{\psi}= a \ket{\psi_1} + \sqrt{1-a^2} \ket{\psi_0}.$$ Given $\epsilon >0$, $\gamma \in ]0,1]$ and access to the oracle $O$, it is possible to produce an estimate $\tilde{a}$ such that $$|\tilde{a} - a| \leq \epsilon, $$
	using $\mathcal{O}(\epsilon^{-1}  \log (1/\gamma))$ queries to $O$, with probability at least $1-\gamma$.
\end{theorem}

Now consider a Boolean function $f$ with domain $\mathcal{T}$. Suppose that $U$ prepares an equal superposition of all elements in $\mathcal{T}$, and let the $1$-subspace be spanned by the states encoding the $1$-inputs of $f$; the corresponding amplitude is $a=\sqrt{t/N}$. Amplitude estimation can be applied to count the number $t$ of $1$-inputs.
 \textcite{Brassard_2002} conceived a counting algorithm based on their approach to QAE. It is possible to adapt their counting algorithm to the alternative formulation we have described. We provide a description of an adaptation and its pseudocode in Appendix~\ref{app:exact_count}, proving that it leads to the same correctness and complexity guarantees.  

\begin{theorem}[Quantum exact counting,~\cite{Brassard_2002}]
	\label{cor:quantum_counting}
	Consider $\mathcal{T}  = \{1,\dots,N\}$, let $f:\mathcal{T}  \to \{0,1\}$, and $t = |\{i \in \mathcal{T} : f(i) = 1 \}|$. 
	Given $\gamma \in ]0,1]$ and access to the oracle $O_f$, it is possible to produce an estimate $\tilde{t}$ satisfying $$|\tilde{t}-t| < \frac{1}{2},$$ such that the integer closest to $\tilde{t}$ is exactly $t$ using $\mathcal{O}(\sqrt{(t+1)(N-t+1)}\log(1/\gamma))$ quantum queries to $f$, with probability at least $1-\gamma$.
\end{theorem}
	
\section{Quantum primitives for clustering}
\label{sec:qalgo}

From the quantum algorithmic standpoint, our key technical novelty is the estimation of densities with a quantum Monte Carlo method (based on amplitude amplification).
Combined with the quantum method for searching nearest-highers by~\textcite{Magano_2023}, these two primitives open the way to quantum speedups for clustering, decision clustering, and cluster counting.

\subsection{Nearest-higher search}
\label{sec:nearest_higher_search}
\textcite{Magano_2023} observe that finding the nearest-higher of an element $x_i \in S$ is equivalent to minimizing the function
\begin{equation}
	\label{eq:function}
	f_i(x_j) = 
	\begin{cases}
		d(x_i,x_j), \text{ if } \rho_j > \rho_i,\\
		+ \infty, \text{ if } \rho_j \leq \rho_i.\\
	\end{cases}
\end{equation}
If the density of all points have been evaluated beforehand, each evaluation of $f_i$ involves a single distance query.
It follows from Theorem~\ref{thm:qmf} that we can find the correct nearest-higher of $x_i$ with probability at least $1-\gamma$ with further
\begin{equation}
	\mathcal{O}\left(\sqrt{N}\log(\frac{1}{\gamma})\right)
	\label{eq:NHscaling}
\end{equation}
distance queries. 
If we can evaluate all densities in time $T$, we can find the $M$ nearest-highers with constant probability of success in query time
\begin{equation}
	\mathcal{O}\left(T + \sqrt{N}M\log(M)\right).
\end{equation}

\subsection{Density estimation}

Algorithms such as those by \textcite{montanaro_quantum_2015} accelerate classical Monte Carlo methods.
Roughly speaking, these algorithms involve encoding the quantity to estimate in the amplitude of some marked state, so that we can use QAE~\cite{AR_2020_qae} to produce the desired estimate.

Let $U_{\chi}$ be a unitary transformation that performs a coherent evaluation of the kernel function,
\begin{equation}
	\label{eq:u_chi}
	U_\chi \ket{b} \ket{d(x_i,x_j)}  = \ket{b \oplus \chi_{ij}} \ket{d(x_i,x_j)},
\end{equation}
where $\chi_{ij} \coloneqq \chi(d(x_i, x_j) / d_c)$.
Here, $b$ is a bitstring of length $n_\chi$, the number of bits required to store the kernel value up to the desired accuracy. 
Such a unitary exists since any classical computation can be emulated in a quantum computer~\cite{Nielsen_Chuang_2010}.
We also introduce the controlled rotation $R_\chi$ acting as
\begin{equation}
	\label{eq:r_chi}
	R_\chi \ket{0}\ket{\chi_{ij}} =  \left(\chi_{ij} \ket{0} + \sqrt{1-\chi_{ij}^2}\ket{1}\right)\ket{\chi_{ij}}.
\end{equation}
Note that this is always well-defined because we assumed previously that $\chi_{ij} \leq 1$.
	
Composing $U_{\chi}$, $R_\chi$, and the distance oracle $O_d$  as in the circuit of  Fig.~\ref{fig:BE_circuit} we implement an operator $U_\rho$ that encodes the density of an element $x_i \in S$ in amplitude as
\begin{equation}
	\label{eq:u_rho}
	U_\rho \ket{0^{n_\rho}}\ket{i}  = \left(\rho_i \ket{0^{n_\rho}}  + \sqrt{1-\rho_i^2}\ket{\perp}\right) \ket{i}.
\end{equation}
$n_\rho = n + n_\chi + n_d + 1$ is the number of ancilla qubits required to perform the operation, and the state $\ket{\perp}$ is  orthogonal to $\ket{0^{n_\rho}}$. 
A detailed calculation is provided in Appendix \ref{sec:appA}, where we establish the following claim.
\begin{claim}
	\label{claim_density}
	The quantum circuit in Fig.~\ref{fig:BE_circuit} implements the $(n_\rho+n)$-qubit unitary $U_\rho$ acting as in Eq.~(\ref{eq:u_rho}), using $\mathcal{O}(\log N)$ single qubit unitary gates and $\mathcal{O}(1)$ distance queries.
	
\end{claim}

\begin{figure*}
\begin{quantikz}[row sep=0.2cm, column sep=0.5cm]
	\lstick{$\ket{0}$} & & &   &   \gate[wires=2][1cm]{R_\chi } &  &  &  & \\ 
	\lstick{$\ket{0^{ n_\chi}}$} &   &  & \gate[wires=2][1cm]{U_\chi } & &   \gate[wires=2][1cm]{U_\chi } &  &  & \\
	\lstick{$\ket{0^{ n_d}}$} &
	&   \gate[wires=3][1cm]{O_d} &    & &  & \gate[wires=3][1cm]{O_d} &  & \\
	\lstick{$\ket{0^{n}}$} &
	\gate{H
	^{\otimes n}}   & &  &   &  &  & \gate{H
	^{\otimes n}} & \qw \\
	\lstick{$\ket{i}$} 
	&  &  &  &    &  &  &  & 
\end{quantikz}
\caption{Quantum circuit implementing $U_\rho$ making use of the Walsh-Hadamard transform $H^{\otimes n}$, the custom transformations $U_{\chi}$ and $R_\chi$ defined in Eq.~\eqref{eq:u_chi} and Eq.~\eqref{eq:r_chi}, respectively, and the distance oracle $O_d$. See Appendix \ref{sec:appA} for a detailed analysis of the circuit.}
\label{fig:BE_circuit}
\end{figure*}
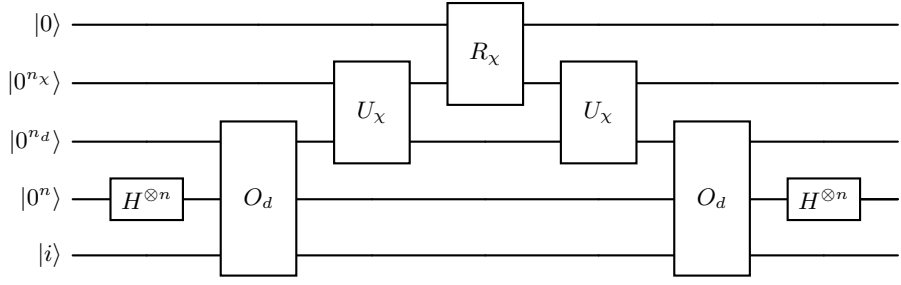

The operator $U_\rho$ works as the state-preparing unitary used in amplitude estimation. 
We can also implement the oracle that marks the state whose amplitude is the density; it is written as
\begin{equation}
	O = (I_{n_\rho} - 2\dyad{0^{n_\rho}}{0^{n_\rho}}) \otimes I_n,
\end{equation}
where $I_{n_\rho}$ is the identity operator acting on the ancilla subspace, and $I_n$ is the identity operator acting on the register that carries the element label. 

It follows from Theorem~\ref{thm:qae} that we can use QAE to estimate the density $\rho_i$ of any element $x_i \in S$ up to an error $\epsilon$ and with probability at least $1-\gamma$ using
\begin{equation}
	\mathcal{O}\left(\frac{1}{\epsilon} \log(\frac{1}{\gamma})\right)
\end{equation}
distance queries.
This constitutes a quadratic speedup in terms of the error dependency when compared with the classical Monte Carlo approach.
The query complexity for estimating all densities, each within error $\epsilon$, with constant probability is
\begin{equation}
	\mathcal{O}\left(\frac{N}{\epsilon} \log(N)\right).
\end{equation}
	
\section{Quantum speedup for clustering}
\label{sec:app}

We now show how to combine the quantum nearest-higher search and the density estimation routines presented in the previous section to reach a speedup for clustering.
We also explain how these primitives benefit three variants of the clustering problem: decision clustering, cluster counting, and heavy cluster counting.
	
\subsection{Clustering}

After assigning a density to every point, DPC requires finding all the $N$ nearest-highers. 
Combining our quantum primitives, it is straightforward to see that we can recover the output of DPC (with constant success probability) using
\begin{equation}
	\mathcal{O}\left(\left(\frac{N}{\epsilon} + N^{3/2} \right)\log(N)\right)
	\label{eq:full_clustering_complexity}
\end{equation}
queries, where each density is evaluated with precision $\epsilon$.

Whether our algorithm constitutes a speedup over the original DPC depends on the required $\epsilon$ for an accurate cluster reconstruction. In fact, the advantage can be achieved if
\begin{equation}
	\label{eq:speedup_condition}
	\epsilon_\mathrm{crit} = \omega \left( \frac{\log N}{N} \right).
\end{equation}
In Section \ref{sec:num} we will see that, while such  error scaling does not hold for all datasets, it seems to describe many practical scenarios.

\subsection{Decision clustering}

The quantum routines for density estimation and nearest-higher search can also provide a quantum speedup for the decision version of the clustering problem.
Given two dataset elements $x_i, x_j \in S$, we want to determine whether or not they belong to the same cluster (without necessarily reconstructing the full clusters).
	
This problem was first addressed in the context of DPC by~\textcite{Magano_2023}.
The key observation is that $x_i$ and $x_j$ belong to the same cluster if and only if the paths starting from them, and formed by repeatedly following the nearest-highers, culminate in the same root. 
Essentially, we need to keep calculating nearest-highers, starting from $x_i$ and $x_j$, until we find the corresponding roots (conditions in Eq.~\eqref{root} are used to check whether a root has been reached). We lay down the detailed steps of an algorithm for decision clustering in Appendix~\ref{app:decisionclustering}.
This approach does not bring any complexity advantage in the classical setting. To see this, note that finding the nearest-higher of an element $x_i$ requires knowing how the density of all other points compares to $\rho_i$. It is unavoidable to compute all the densities, which costs $\mathcal{O}(N^2)$ distance queries.
Nevertheless,~\textcite{Magano_2023} showed how to overcome the quadratic scaling with quantum nearest-higher search. They have defined $H$ as the maximum length of any nearest-higher path culminating in a root; in the worst case, one needs to use the nearest-higher search routine $2H$ times before reaching the root for both inputs. Densities can be computed exactly in the course of the nearest-higher searches, as they are required in each evaluation of functions like $f_i$ (Eq.~\eqref{eq:function}). As a result, the quantum query complexity for solving the decision clustering with constant probability of success is in $\mathcal{O}(N^{3/2}H \log(H))$.

Here, we propose replacing the exact density computation by the quantum density estimation. 
However, the stochastic density evaluation raises some problems that were not present in the original version of DPC.
If we rely on a lazy density evaluation for nearest-higher search (as in the original proposal \cite{Magano_2023}), we will  produce potentially different density estimates for the same dataset element. 
In general, the configuration of nearest-highers and the definition of roots change whenever the hierarchy of density estimates changes, even if the original cluster structure is maintained.
Consequently, two nearest-higher paths starting from points of the same cluster could culminate in different roots, leading to an erroneous answer to decision clustering.

One simple alternative is to estimate the density of all elements before the start of the nearest-higher searches.
Afterwards, we build the two nearest-higher paths (demanding constant success probability over the entire paths).
The corresponding algorithm has a quantum query complexity in
\begin{equation}
	\mathcal{O}\left(\frac{N}{\epsilon} \log (N) + \sqrt{N}H \log (H)\right).
	\label{eq:decision_clustering_complexity}
\end{equation}

The condition for quantum speedup is the same for the full clustering algorithm (Eq.~\eqref{eq:speedup_condition}).
Our algorithm improves upon the version of~\textcite{Magano_2023} whenever
\begin{equation}
	H\epsilon_\mathrm{crit} = \omega\left(\frac{\log N }{\sqrt{N}\log (H)}\right).
\end{equation}
In terms of the effective dimension, this can be expressed as $\epsilon_\mathrm{crit} = \omega(N^{-1/2 - 1/d_\mathrm{eff}})$.
For $d_\mathrm{eff} >2$, it is a stronger condition than that expressed in Eq.~\eqref{eq:speedup_condition}, but it can still be satisfied in many practical scenarios, as we will see in Section~\ref{sec:num}. 

Finally, it should be noted that running our decision clustering algorithm does not always compensate for running the full clustering.
Comparing \eqref{eq:decision_clustering_complexity} with \eqref{eq:full_clustering_complexity}, we see that the decision clustering strategy is preferable when
\begin{equation}
	H = o(N),
\end{equation}
or, in terms of the effective dimension, $d_\mathrm{eff} > 1$.

\subsection{Cluster counting}

The second variant that we consider is cluster counting. 
Say that the dataset is divided into clusters $C_1, \ldots, C_K$ according to some ground truth -- in our context, the clusters are defined relative to DPC. 
The goal is to determine the number of clusters, $K$.

Classically, counting the number of clusters with DPC seems to be as hard as the full clustering problem, scaling as $\mathcal{O}(N^2)$.
In contrast, we present a subquadratic quantum algorithm for cluster counting.

First, we observe that DPC identifies each cluster by the respective root, that is, counting the number of clusters is the same as counting the number of roots.
Recall that a root is an element satisfying the conditions established in Eq.~\eqref{root}. Equivalently, a root is a $1$-input of the Boolean function $f_r$ defined by 
\begin{equation}
    f_r(x_i) \coloneqq \left(\delta_i \geq \delta_c\right) \wedge \left(\rho_i \geq \rho_c\right),
\end{equation}
for a given choice of threshold parameters $(\delta_c,\rho_c)$. 
Given access to a root identification operator $O_{f_r}$ acting as
\begin{equation}
	\label{eq:root_oracle}
    O_{f_r}\ket{i} = (-1)^{f_r(x_i)}\ket{i},
\end{equation}
we can employ the quantum exact counting (Theorem~\ref{cor:quantum_counting}) to count the number of $1$-inputs of $f_r$, i.e., the number of roots. 
The quantum cluster counting algorithm is summarized in Appendix~\ref{app:clustercounting}.

To build $O_{f_r}$ we need to evaluate the density $\rho_i$ and the nearest-higher separation $\delta_i$ coherently (and then compare these values with the threshold parameters).
For that purpose, we again pre-compute all the density values and then employ a coherent version of quantum minimum finding.
Since the quantum nearest-higher search is not deterministic, we do not implement the operator defined in Eq.~\eqref{eq:root_oracle} but an approximation $\widetilde{O}_{f_r}$.
In Appendix~\ref{sec:root_oracle}, we provide the details for such an implementation, including the demonstration of the following result.
\begin{claim}
    \label{claim:root_oracle}
    Suppose that densities have been assigned to all dataset elements, and that their values can be accessed coherently. Given $\Delta > 0$, it is possible to implement $\widetilde{O}_{f_r}$ satisfying
    \begin{equation*}
            \bigg|\bigg| \widetilde{O}_{f_r} - O_{f_r}\bigg|\bigg|_2 \leq \Delta,
    \end{equation*} 
    with $\mathcal{O}(\sqrt{N}\log(1/\Delta))$ distance queries.
\end{claim}

The quantum exact counting uses $M = \mathcal{O}(\sqrt{(K+1)(N-K+1)}\log(1/\gamma))$ queries to $O_{f_r}$ and counts the exact number of clusters, with probability at least $1-\gamma$. To keep the same probability of success when working with $\widetilde{O}_{f_r}$, we set $\Delta = \mathcal{O}(\gamma/M)$.
The query complexity of quantum cluster counting follows from combining the cost of implementing $\widetilde{O}_{f_r}$ with the number $M$ of implementations, and also from including the cost for estimating all densities, resulting in
\begin{equation}
	\mathcal{O}\left(\left(\frac{N}{\epsilon}  + \sqrt{N K (N-K + 1)} \right)\log(N)\right)
\end{equation}
distance queries.

The speedup condition is the same as Eq.~\eqref{eq:speedup_condition}.
For a constant number of clusters, if $\epsilon_{\mathrm{crit}} = \omega(1)$ we find a quadratic speedup.
Unlike the case for the decision clustering problem, it is never more efficient to perform the full clustering than to count clusters via our dedicated algorithm.

\subsection{Heavy cluster counting}

We now consider a variation of the cluster counting where we are interested in estimating the number of clusters whose mass exceeds a certain threshold.
This problem is close to that of estimating the number of species that were not yet observed in a given ecosystem -- the so-called unseen species problem. This classic problem was originally covered by~\textcite{fisher1943}, and different versions have been studied subsequently (see, for example, \textcite{efron_estimating_1976}).
Solving this type of problem usually requires extrapolating the number of classes from a finite sample of them. 
We follow the same approach here, observing that DPC allows us to sample clusters by sampling data points uniformly at random and then building the nearest-higher paths up to the respective roots.

Recall that $C_1,\dots, C_K$ are the clusters defined relative to DPC. 
For statistical completeness, define $C_0$ as the spurious cluster composed of all outliers in the dataset. The probability of sampling the cluster $C_k$ is the fraction of data points that belong to it. 
We refer to this probability as the \emph{mass} of the cluster, and we denote it by $m_k$, such that
\begin{equation}
	m_k = \frac{\vert C_k \vert}{N}, k \in \{ 0,\dots, K\}.
\end{equation}

The estimation of the total number of clusters may become prohibitively hard if there are some clusters with very small mass (for example, if $m_k = o(1/N)$ for some $k$). 
In that light, we suggest estimating the number of clusters whose mass exceeds a certain threshold, $m$. The heavy cluster count was defined formally in Section~\ref{sec:intro} (the cluster of outliers, $C_0$, is excluded from the count). 
Clearly, we recover exact counting in the limit $m, \eta = 0$.

Let $\Sigma$ be a random sample of clusters of size $t$ taken with replacement. The mass of the clusters that are not included in $\Sigma$ is called \emph{missing mass}, and it is given by 
\begin{equation}
	M_0(t) = \sum_{k=0}^K m_k \mathds{1}\{C_k \notin \Sigma\},
\end{equation}
where $\mathds{1}\{C_k \notin \Sigma\}$ is the indicator function that evaluates to $1$ when $C_k$ is not in the sample, and to $0$ otherwise.
When $M_0(t) \leq m$, it is guaranteed that all clusters with mass greater than $m$ have already been sampled.
So, we can use the missing mass to establish a halting condition for the sampling process.
For that purpose, we employ the Good-Turing estimator for the missing mass~\cite{good_population_1953}, denoted by $G_0(t)$. 
If $\phi_1(t)$ is the random variable representing the number of clusters that appear exactly once in $\Sigma$, the Good--Turing estimator reads
\begin{equation}
	G_0(t) \coloneqq \frac{\phi_1(t)}{t}.
\end{equation}
The idea is then the following: we first draw a sufficient number of samples to ensure that (i) $M_0 \leq G_0 +m/2$, with high probability. Once in that regime, we keep sampling until we observe that (ii) $G_0 \leq m/2$. From the combination of both inequalities, it follows that $M_0 \leq m$ with high probability. 

Naturally, we do not know beforehand how many samples we need to take before condition (i) is satisfied.
But we can establish a bound on the expected number as in the claim below, which we prove in Appendix~\ref{sec:appC}.
\begin{claim}
\label{claim:6}
Let
\begin{equation}
	\tau_1 \coloneqq \frac{5400}{m^2} \log \left(\frac{2}{m\gamma}\right).
\end{equation}
Then, for $t \geq \tau_1$, 
\begin{equation}
	\mathrm{P}\left[M_{0}(t) > G_{0}(t) + m/2 \right] \leq \frac{\gamma}{4t^2}.
\end{equation}
\end{claim}

Even if the sampling is halted when $M_0 \leq m$, some clusters with mass less than $m$ can be included in $\Sigma$. To exclude these from the count, we estimate the mass of all sampled clusters and discard those for which the estimated mass is not greater than $m$. The Monte Carlo estimate $\hat{m}_k(t)$ for $m_k$ is the frequency with which the cluster $C_k$ was observed, i.e.,
\begin{equation}
	\hat{m}_k(t) \coloneqq \frac{|\{C \in \Sigma : C = C_k\}|}{t}.
\end{equation}
The spurious cluster $C_0$ is excluded since we are not including it in the count.

For this strategy to work we also need to draw a sufficient number of samples to ensure that the estimates are accurate enough. 
For that, we resort to Claim~\ref{claim:5}, proven in Appendix~\ref{sec:appB}.
\begin{claim}
\label{claim:5}
Let
\begin{equation}
	\tau_2 \coloneqq \frac{16}{\eta^2}\log\left(\frac{N}{\eta \gamma}\right).
\end{equation}
Then, for $t \geq \tau_2$, 
\begin{equation}
	\mathrm{P}\left[\bigcup_{k=1}^K\left\{|\hat{m}_k(t)-m_k| \geq \eta\right\}\right] \leq \frac{\gamma}{4t^2}.
\end{equation}
\end{claim}

With these results, we propose the following algorithm to count the heavy clusters. 
We begin by drawing a number $\tau \coloneqq \max \{\tau_1, \tau_2\}$ of samples, such that both Claim~\ref{claim:6} and Claim~\ref{claim:5} hold.
After that, we continue the sampling process until $G_0\leq m/2$.
Then, we output the estimate
\begin{equation}
	\widehat{K_m} \coloneqq |\{k=1,\dots,K: \hat{m}_k > m\}|
\end{equation}
The steps are detailed in Algorithm~\ref{algo:heavycounting}, which can be found in Appendix~\ref{app:heavycounting}. 
The following result, proven in Appendix~\ref{sec:appD}, establishes the correctness and the complexity of the algorithm.

\begin{claim}
\label{claim:count}
Let $K$ be the number of clusters, $m > \eta > 0$, and $\gamma \in ]0,1]$. Algorithm~\ref{algo:heavycounting} returns $\widehat{K_m}$ such that 
\begin{equation}
	K_{m+\eta} \leq \widehat{K_m} \leq K_{m-\eta},
\end{equation}
with probability at least $1-\gamma$, using
\begin{equation}
	\mathcal{O}\left(\frac{K}{m} + \frac{1}{\eta^{2}}\log\frac{N}{\eta \gamma}\right)
\end{equation}
cluster samples.
\end{claim}

The query complexity of the algorithm follows directly from adding the cost of estimating all densities beforehand to the cost of building the nearest-higher path for each sample, resulting in
\begin{equation}
	\label{eq:speedup_heavy_counting}
	\mathcal{O}\left(\frac{N}{\epsilon} \log(N) + \sqrt{N}H\left(\frac{K}{m} + \frac{1}{\eta^2}\log(\frac{N}{\eta})\right)\log(\frac{HK}{m\eta})\right)
\end{equation}

For quantum speedup, $\epsilon_{\mathrm{crit}}$ is subject to the same condition as before (Eq.~\eqref{eq:speedup_condition}).
In addition, we require that 
\begin{equation}
	H =  o\left(\frac{N^{3/2}}{\log^2(N /m \eta)} \,  \min\left( \frac{m}{K}, \eta^2 \right) \right).
\end{equation}
For constant $K$, $m$, and $\eta$ this condition is trivially met. 

If $m$ and $\eta$ are constant, the heavy cluster counting strategy beats the exact counting one if
\begin{equation}
	H = o\left( \sqrt{\frac{N}{K}} \frac{1}{\log (NK)} \right).
\end{equation}
Otherwise, it is more efficient to perform exact counting.

\section{Critical error}
\label{sec:num}

The error introduced in local densities when these are estimated may alter the cluster structure induced by DPC. This is because the clusters built by DPC are based on the density hierarchy in the dataset: each point is assigned to the same cluster as its nearest-higher, which is the nearest neighbor among the points that are above it in the hierarchy.
Errors of the same magnitude as the gaps between distinct densities can cause changes in the original hierarchy, which in turn modify the nearest-higher relationships and potentially lead to modified clusters.
 
Recall the definition of critical error, $\epsilon_\mathrm{crit}$, as the threshold error from which the cluster structure begins to be modified. 
We must have that
\begin{equation}
	\epsilon_{\mathrm{crit}} \geq \frac{1}{2}\min_{\substack{x_i,x_j \in S \\ \rho(x_i) \neq \rho(x_j)}} \big |\rho_i - \rho_j \big |,
	\label{eq:crit_lower_bound}
\end{equation}
since below this lower bound the density hierarchy is unchanged. 
This is a conservative lower bound: most reorderings in the hierarchy are harmless, because they reroute a point to a different nearest-higher within the same cluster.

We have seen that our quantum algorithms outperform classical DPC depending on the same condition for $\epsilon_\mathrm{crit}$, Eq.~\eqref{eq:speedup_condition}.
We analyze whether this condition can be met in practice.

\subsection{A toy dataset}

\begin{figure*}
		\centering
		\begin{subfigure}[b]{0.45\textwidth}
			\centering
			\includegraphics[width=\linewidth]{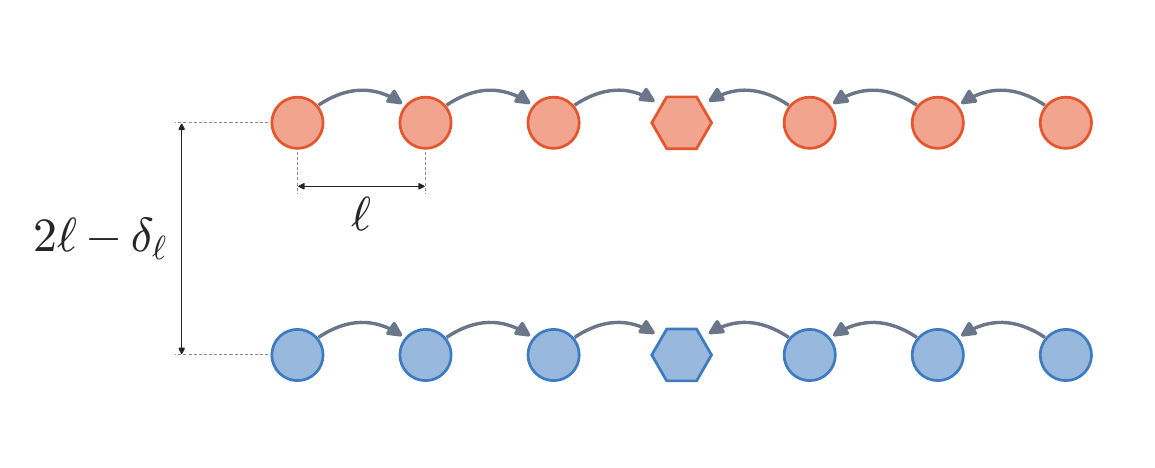}
			\caption{\label{fig:toy_a}}
		\end{subfigure}
		\begin{subfigure}[b]{0.45\textwidth}
			\centering
			\includegraphics[width=\linewidth]{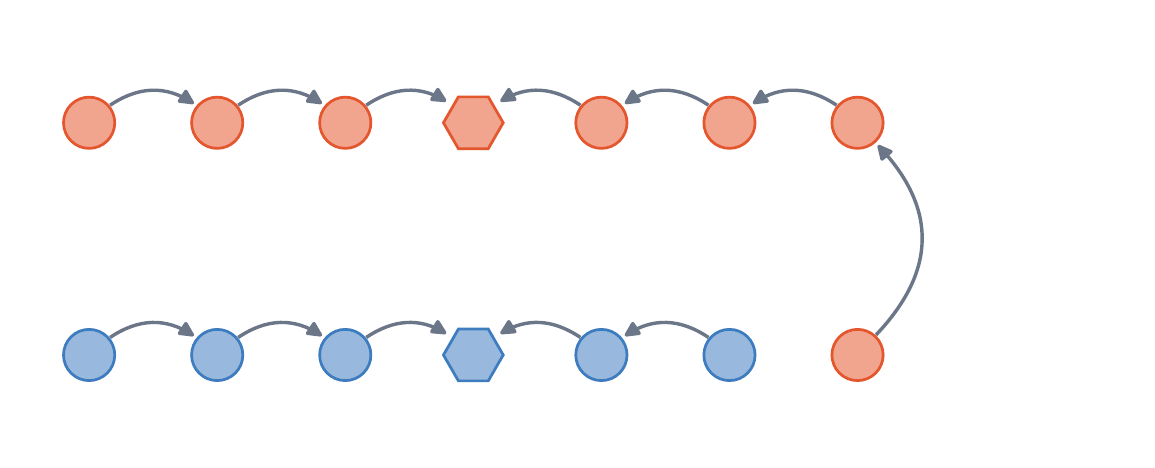}
			\caption{\label{fig:toy_b}}
		\end{subfigure}
		\caption{\textbf{Impact of errors on density peak clustering}. Toy dataset consisting of $N=14$ points equally divided into two horizontal clusters, separated by a distance of $2\ell - \delta_\ell$, with $\delta_\ell \ll \ell$. Colors indicate the cluster assigned by DPC, roots are represented as hexagons, and arrows connect each remaining point to its nearest-higher. \textbf{(a)} Clusters found by DPC when operating with exact densities. \textbf{(b)} Clusters found by DPC when operating with density estimates of error $\epsilon = 2/N$: the rightmost point of the lower cluster is assigned to the upper one.}
		\label{fig:toy_example}
\end{figure*}

Consider the toy dataset illustrated in Fig.~\ref{fig:toy_example}.
The dataset is embedded in $\mathbb{R}^2$ and consists of $N$ (even) points divided into two clusters of the same size. In each cluster, the points are evenly distributed along a horizontal line, and we take $N/2$ to be odd, so that each cluster has a single central point. The distance between neighboring points is $\ell$, and the length of each cluster is $L = (N/2 - 1)\ell$. The two clusters are separated vertically by a distance of $2\ell - \delta_\ell$, with $\delta_\ell \ll \ell$. The distance function is the infinity norm and densities are defined with the Heaviside step kernel with $d_c = L/2$. 

The central points in each cluster have density equal to $1$. 
The remaining densities decrease in increments of $2/N$ as the points move away from the center.
For an adequate choice of threshold parameters (for example, $\rho_c = 1-2/N$ and $\delta_c = \ell$) DPC classifies each central point as a root and correctly identifies both clusters, as illustrated in Fig.~\ref{fig:toy_a}.

Now consider that densities are estimated within error equal to the minimum density gap, i.e., $\epsilon = 2/N$. As a consequence, a given point may have a density that is higher than the density of its neighbor closer to the center of the cluster. In that scenario, its nearest neighbor of higher density within the same cluster is now at a distance of at least $2\ell$. There might be a point in the opposite cluster with higher density, which is closer by the small amount $\delta_\ell$. This is sufficient for that given point to be wrongly assigned to the opposite cluster, as illustrated in  Fig.~\ref{fig:toy_b}.
We conclude that $\epsilon_{\mathrm{crit}} = \mathcal{O}(1/N)$ and no speedup is possible for any of our algorithms.

\begin{figure}
	\centering
	\includegraphics[width=0.65\linewidth]{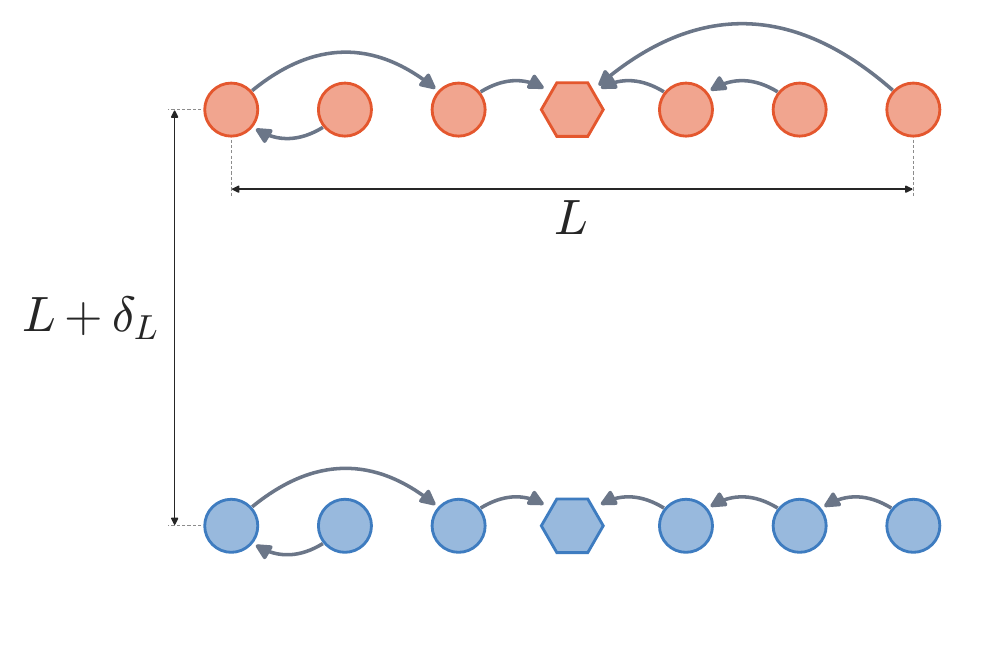}
	\caption{\textbf{Absence of critical error}. A variant of the toy dataset presented above in which the distance between clusters is $L+\delta_L$, with $\delta_L > 0$. DPC identifies the correct clusters when operating with density estimates of constant error, $\epsilon = \Theta(1)$.}
	\label{fig:toy_example2}
\end{figure}

Evidently, this artificial example was specifically designed so that the speedup condition would not be met.
In contrast, suppose that the vertical distance between clusters is $L + \delta_L$, for some $\delta_L > 0$, as illustrated in Fig.~\ref{fig:toy_example2}.
In that case, since the distance between clusters exceeds the length of each cluster, the nearest-higher of every point is a point within the same cluster, no matter how large we set the error to be. 
In other words, $\epsilon_{\mathrm{crit}} = \Omega(1)$.

The second scenario would generally be considered the easiest for clustering, in the sense that the wide separation between clusters should favor their identification, irrespective of the method adopted.
This suggests that $\epsilon_{\mathrm{crit}}$ could be seen as a measure of clusterability (compare, for example, with the concept of conductance introduced in Ref.~\cite{czumaj}).
We can say that the cluster structure of the dataset in this second case is stronger than that in the original toy dataset -- the clusters are the same but they are widely separated, which favors their identification.

We corroborate our analysis with numerical experiments.
To evaluate the fidelity of the cluster partition obtained with finite-error density estimates relative to that induced by the original DPC, we make use of two measures known as Adjusted Rand Index (ARI)~\cite{hubert_comparing_1985} and Adjusted Mutual Information (AMI)~\cite{vinh_information_nodate}.
Both quantities assume the maximum value $1$ when the two cluster partitions match perfectly, and they approach $0$ if the matching between partitions is no better than would be obtained if both partitions were selected at random (subject to having the same number of clusters and elements per cluster in each).

We introduce an error in the density estimates in a way that roughly simulates how it arises in quantum amplitude estimation (see Section~\ref{sec:pre}). The exact densities are first transformed into phases, $\theta_i \coloneqq \arcsin(\rho_i)$; we obtain $\tilde{\theta}_i$ by adding an error drawn from a Gaussian $\mathcal{N}(0,\epsilon^2)$; the resulting densities are given by $\tilde{\rho}_i \coloneqq \sin \tilde{\theta}_i $. In each simulation, we perform eight independent runs of DPC with erroneous densities, for different values of $\epsilon$; in each run, we compute the $\mathrm{ARI}$ and the $\mathrm{AMI}$ between the resulting clustering and that of original DPC; we average the results over the different runs.

\begin{figure}
		\centering
		\includegraphics[width=\linewidth]{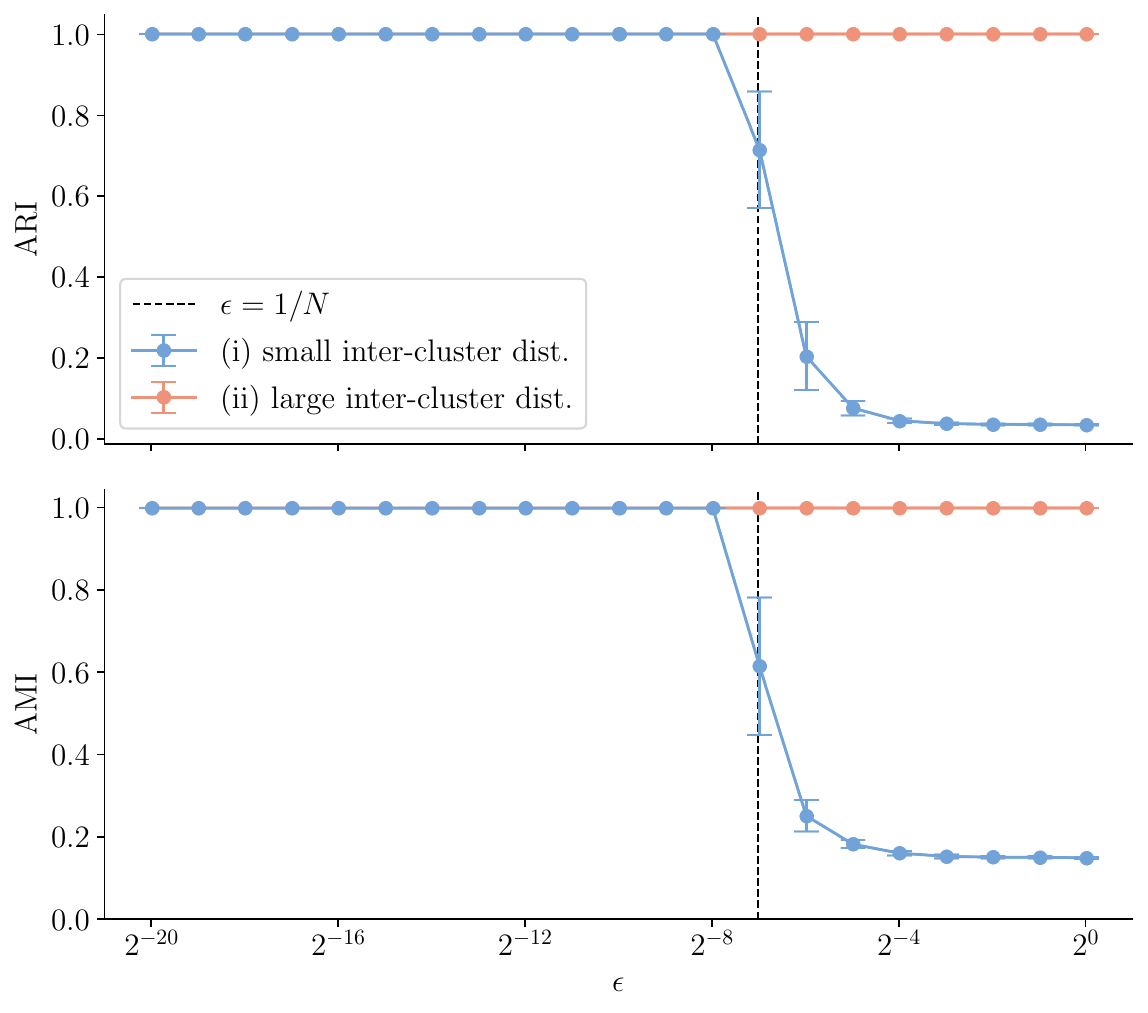}
		\caption{\textbf{Clustering fidelity indices versus $\epsilon$ for the toy dataset}. The top plot shows the value of $\mathrm{ARI}$ as a function of $\epsilon$; the bottom one shows the value of $\mathrm{AMI}$ as a function of $\epsilon$. The blue curve refers to case (i), where clusters are separated by $2\ell - \delta_\ell$, with $\delta_\ell \ll \ell$; the red curve refers to case (ii), where clusters are separated by $L + \delta_L$, with $\delta_L > 0$.}
		\label{fig:toy_acc}
\end{figure}

We consider the two versions of the toy dataset with $N=258$ points: (i) clusters separated by $2\ell - \delta_\ell$ with $\delta_\ell \ll \ell$; (ii) clusters separated by $L + \delta_L$, with $\delta_L > 0$. Fig.~\ref{fig:toy_acc} displays the value of $\mathrm{ARI}$ and $\mathrm{AMI}$ as a function of $\epsilon$.
The plot is consistent with our previous analysis. When clusters are close to each other, there is a critical error beyond which DPC starts to degrade. The vertical dotted line crosses the value of $\epsilon$ from which both indices fall below $0.9$, providing an approximation to $\epsilon_\mathrm{crit}$.
 On the other hand, when the clusters are well-separated from each other, DPC recovers the clusters perfectly even with errors of the same magnitude as the densities themselves.

\subsection{Common benchmarks}

Now we study some common datasets, finding evidence that the critical error does not have a strong dependence on the size of the dataset.
We consider three different types of dataset (the first is a synthetic dataset, while the other two are composed of real data). For each, we build an exemplar dataset with a given number of points. By sampling elements uniformly at random from that exemplar, we generate sub-datasets of smaller sizes. 
The elements belonging to a given sub-dataset are also included in the sub-datasets of larger size. The datasets considered are the following:

\vspace{0.2cm}

\noindent \hyperref[fig:gaussians_N]{(a)} \textit{Two Gaussian clusters in $\mathbb{R}^2$}. Two centroids were selected in $\mathbb{R}^2$ at a distance $2a$ from each other. We centered one isotropic Gaussian distribution on each centroid, with variance $\sigma^2 = (\pi/8)a^2$. This variance ensures that the clusters do not overlap substantially: it makes the mean absolute deviation of a point from its centroid, measured along the direction connecting both centroids, equal to $a/2$, i.e., half the distance from the centroid to the midpoint between clusters. The exemplar dataset consists of $N = 16384$ points sampled from those Gaussian distributions.

\vspace{0.2cm}

\noindent \hyperref[fig:urbanGB]{(b)} \textit{Urban road accidents in Great Britain}. The original dataset has $360\,177$ elements, each with two features corresponding to the latitude and longitude of urban accidents that occurred in Great Britain during a specific period~\cite{urbangb}. After multiplying the longitudes by the cosine of the mean latitude, the $2$-norm is used as the distance function, since, at this scale, it is a good approximation to the geodesic distance.
We filtered the original dataset in order to include only the clusters with more than $3100$ elements, such that all clusters are represented by at least $10$ points in all sub-datasets with high probability. The resulting exemplar dataset has $N=24173$ points.

\vspace{0.2cm}

\noindent \hyperref[fig:pendigits]{(c)} \textit{Handwritten digits}. The dataset has $10\,992$ elements, each having $16$ features corresponding to the coordinates of $8$ points sampled from the pen's trajectory while writing digits~\cite{pendigits}. Although this dataset does not have a straightforward geometrical interpretation, we adopt the $2$-norm in $\mathbb{R}^{16}$ as the distance function. We took the original dataset as the exemplar, in this case.

\vspace{0.2cm}

We perform the simulation for different values of $N$, showing the results in Fig.~\ref{fig:numerical_1}. 
The curves display the same qualitative behavior for the three datasets and for all sizes. 
For small errors, the clustering induced by DPC is unchanged. 
But, as the error gets larger, the clustering performance becomes significantly worse.
For the second dataset \hyperref[fig:urbanGB]{(b)} the result approaches a random clustering.
For the other two datasets both indices stabilize at a value distinct from zero, indicating that some of the original cluster structure is still recovered.

We plot the value of $\epsilon$ for which either index falls below $0.9$, constituting an approximation of $\epsilon_\mathrm{crit}$. 
We see that, for all three datasets, $\epsilon_\mathrm{crit}$ remains roughly constant as we change $N$.
This suggests that these datasets meet our condition for quantum speedup (Eq.~\eqref{eq:speedup_condition}).

\begin{figure*}
		\centering
		\begin{subfigure}[b]{0.32\textwidth}
			\centering
			\includegraphics[width=\linewidth]{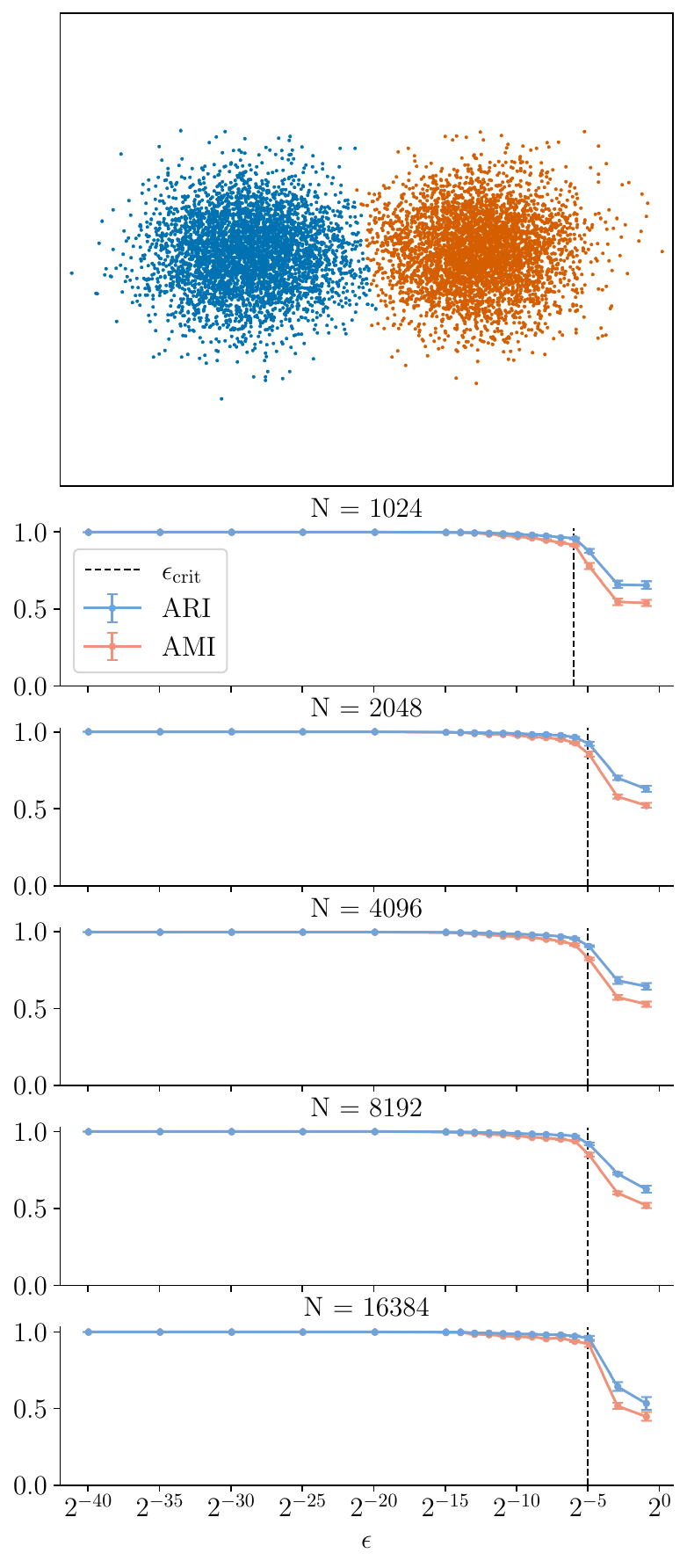}
			\caption{\label{fig:gaussians_N} }
		\end{subfigure}
		\begin{subfigure}[b]{0.32\textwidth}
			\centering
			\includegraphics[width=\linewidth]{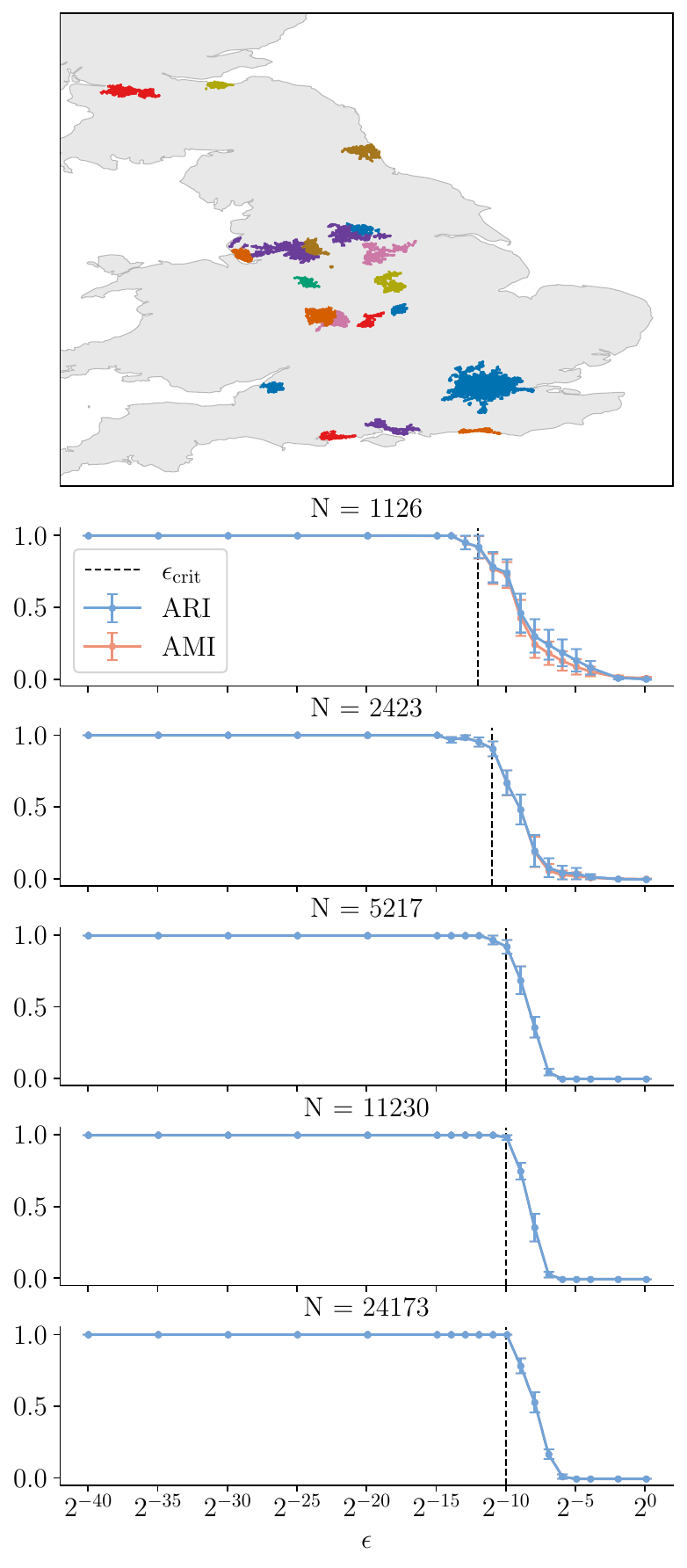}
			\caption{\label{fig:urbanGB} }
		\end{subfigure}
		\begin{subfigure}[b]{0.32\textwidth}
			\centering
			\includegraphics[width=\linewidth]{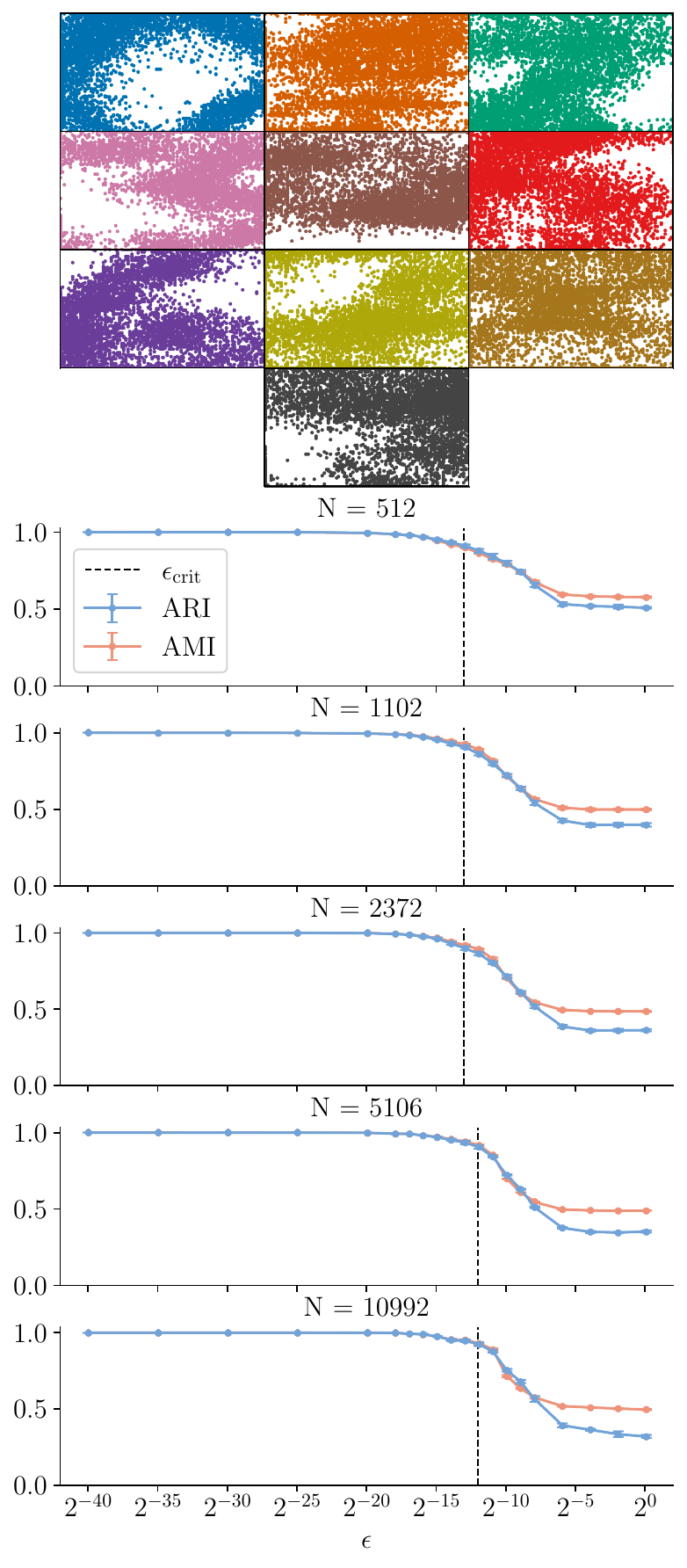}
			\caption{\label{fig:pendigits}}
		\end{subfigure}
		\caption{\textbf{Independence of the critical error from size}. The plots show the clustering fidelity indices (ARI in blue, AMI in red) as a function of $\epsilon$, for \textbf{(a)} two Gaussian clusters in $\mathbb{R}^2$, \textbf{(b)} urban road accidents in Great Britain, and \textbf{(c)} handwritten digits.}
		\label{fig:numerical_1}
\end{figure*}

\begin{figure}
		\centering
		\includegraphics[width=0.9\columnwidth]{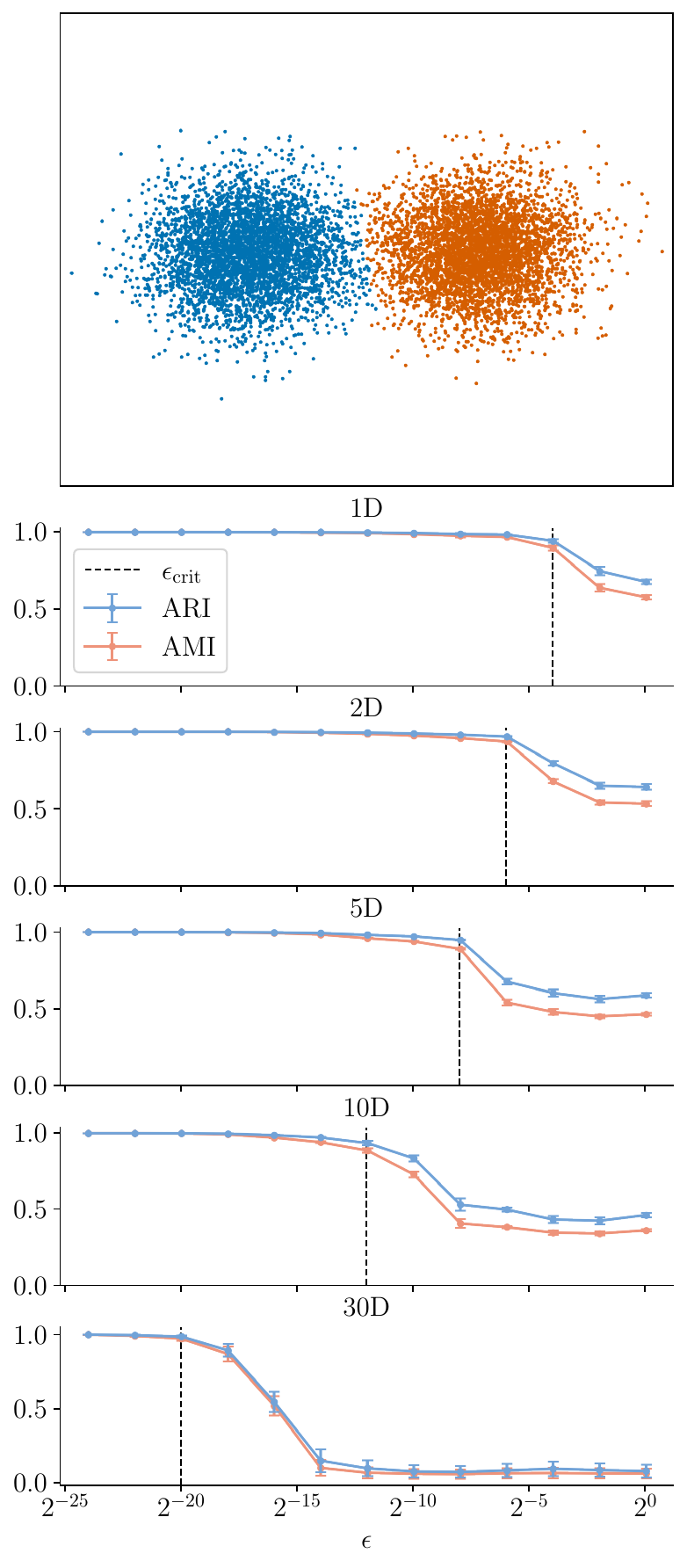}
		\caption{\textbf{Dependence of the critical error on dimensionality}. Clustering fidelity indices (ARI in blue, AMI in red) as a function of  $\epsilon$, for two Gaussian clusters in $\mathbb{R}^D$, with $D \in \{1,2,5,10,30\}$.}
		\label{fig:gaussians_D}
\end{figure}

\subsection{Varying dimension}

In our last numerical experiment we explore the dependence of the critical error on the dimensionality of the dataset.
Consider the synthetic dataset composed of two Gaussian clusters, with centroids at the same distance, and the same variance $\sigma^2$. 
This time, we fix the number of elements at $N=2^{13}$, and we vary the dimensionality $D$ of the Euclidean space in which the clusters are embedded. 
The expected distance between points sampled from the same Gaussian scales asymptotically with $\sqrt{D}$. 
In order to maintain the same reliability in the exact densities in all dimensions, we set the scale parameter to have the same scaling, i.e., $d_c(D) = d_{c,0}\sqrt{D/2}$, where $d_{c,0}$ is the reference value in $2$ dimensions. 
The results of the clustering are shown in Fig.~\ref{fig:gaussians_D}.

We observe a clear dependence of the critical error on the dimensionality of this dataset: as $D$ increases, the critical error gets smaller.
We view this as a manifestation of the `curse of dimensionality', a term generally used to describe the panoply of phenomena that make data analysis challenging in high-dimensional spaces~\cite{bellman1961}.
In this case, we observe the concentration of the norm around its expected value as the dimension of the space grows.
The generated dataset elements concentrate on a shell of radius $\sigma\sqrt{D}$ and constant thickness $\Theta(\sigma)$ \cite{Vershynin_2018}. 

Two difficulties arise from this phenomenon. 
The first is that, in high dimensions, the clusters cease to resemble Gaussian blobs and begin to look more like shells.
Almost no sampled point lies near the centroid, so that it is harder to identify a density peak in each cluster. 
Consequently, the cluster structure of the dataset becomes weaker, making it more difficult for DPC to detect the clusters. 
The second is that, as the dimension grows, the inter-cluster distance becomes harder to distinguish from the intra-cluster distance. 
To see this, note that the vector connecting two random points from different clusters follows a Gaussian distribution, with mean vector equal to the vector connecting the two centroids, and variance $2\sigma^2$. 
Due to the concentration of the norm, the distance between such a pair of points concentrates around $\sqrt{4a^2+2\sigma^2 D}$, with fluctuations of order $\Theta(\sigma)$. 
Applying the same argument to a pair of random points from the same cluster gives a distance concentrated around $\sigma\sqrt{2D}$, with fluctuations of the same order. In the limit $\sigma^2D \gg a^2$, the difference between the inter-cluster and the intra-cluster distances scales as $\sqrt{2}a^2/(\sigma\sqrt{D})$. Therefore, for large $D$, the two types of distance become statistically indistinguishable. Each point is approximately equally distant from points in its own cluster and from points in the other cluster. Since all distances become similar, densities approach one another and a minor error is sufficient to modify the density hierarchy.

To sum up our analysis, for some practically relevant cluster structures $\epsilon_\mathrm{crit}$ depends weakly on the number $N$ of points in the dataset, remaining approximately constant over the range of $N$ considered. 
In those cases, our speedup condition holds and our algorithms provide an advantage over classical DPC.
In contrast, $\epsilon_\mathrm{crit}$ seems to be quite sensitive to the intrinsic dimension of the dataset.
Looking at a toy dataset, we managed to build a scenario where $\epsilon_\mathrm{crit} = \mathcal{O}(1/N)$ and another where $\epsilon_\mathrm{crit} = \omega(1)$, showing that no limit can be excluded \textit{a priori}.

\section{Conclusion}
\label{sec:conclusion}

In this work, we developed a quantum approach to density peak clustering (DPC) by combining quantum Monte Carlo techniques with quantum search primitives. Using this approach, we proposed algorithms for four problems: 
\begin{enumerate}
	\item Clustering: partition the dataset into clusters.
	\item Decision clustering: decide if two given elements belong to the same cluster.
	\item Cluster counting: determine the number of clusters in the dataset.
	\item Heavy cluster counting: estimate the number of clusters with mass above a certain threshold.
\end{enumerate}
Their query complexities are summarized in Table~\ref{tab:2}, and detailed pseudocode is provided in Appendices~\ref{app:clustering}, \ref{app:decisionclustering}, \ref{app:clustercounting}, and \ref{app:heavycounting}.

The original DPC performs clustering with complexity $\mathcal{O}(N^2)$, and we claim that the same quadratic scaling is unavoidable in the solution for the other three problems. 
For clustering, decision clustering, and cluster counting, we demonstrate a polynomial quantum speedup under a condition on the critical error, $\epsilon_\mathrm{crit}$.
The algorithm developed for heavy cluster counting also provides a quantum speedup with an additional condition on the dataset parameter $H$.

We constructed a family of toy datasets for which $\epsilon_\mathrm{crit}=\mathcal{O}(1/N)$, showing that our approach cannot lead to an unconditional speedup. A simple modification of this construction instead gives $\epsilon_\mathrm{crit}=\Omega(1)$, allowing polynomial speedups when the remaining parameters also scale favorably.

Our numerical experiments further suggest that, for several practically relevant datasets, $\epsilon_\mathrm{crit}$ depends only weakly on $N$ and remains approximately constant over the range of sizes considered. Instead, it appears to be more sensitive to structural properties of the data, particularly cluster separation and intrinsic dimensionality. A complete characterization of this dependence lies beyond the scope of the present work. Establishing such a characterization would help identify the settings in which quantum DPC is advantageous and could contribute more broadly to understanding how quantum speedups in machine learning depend on data geometry.

Although it has not been the focus of this work, we also note that the Monte Carlo approach to density estimation suggests alternative classical algorithms for decision clustering, cluster counting, and heavy cluster counting.
Depending on the scaling of the parameters $\epsilon$ and $H$, these algorithms can be more efficient than the original DPC (but not more efficient than our quantum algorithms).
Unfortunately, the Monte Carlo approach alone cannot lead to a more efficient classical algorithm for the full clustering problem.

It remains open whether an unconditional quantum speedup can be obtained for DPC in the black-box model. Such a result would require a quantum algorithm with provably subquadratic query complexity without assumptions on the structural properties of the dataset.

\section{Generative AI Disclosure}
We used generative AI models to help conceptualize adaptations for amplitude estimation and exact counting primitives, and to revise the language. We have reviewed all content and take full responsibility for it.

\section{Acknowledgments}

We acknowledge support from FCT-Portugal through Grant No. UID/04650/2025.

\clearpage
\appendix

\onecolumngrid

\section{Pseudocode for the adapted quantum amplitude estimation \label{app:additive_error}}

\begin{algorithm}
\SetAlgoLined
\KwIn{$\epsilon, \gamma \in ]0,1]$, access to $U$ and $U^\dagger$, and to the oracle $O$.}
\KwOut{An estimate $\tilde{a}$ for $a$ within error $\epsilon$}
\BlankLine

\vspace{0.2cm}
\tcp{Let $R$ be the single-qubit rotation around the $\mathrm{Y}$-axis through an angle $2\arccos(1/c)$, with $c=1001$. Let $\ket{\psi'}$ be the $(n+1)$-qubit state prepared by $U' \coloneqq U \otimes R$, $$\ket{\psi'} = \frac{a}{c}\ket{\psi_1}\ket{0} + \sqrt{1-\left(\frac{a}{c}\right)^2}\ket{\perp},$$ where $\ket{\perp}$ is orthogonal to $\ket{\psi_1}\ket{0}$ (which redefines the $1$-subspace). Let $O'$ be the oracle implemented using $O$ controlled on the second register. The algorithm runs on this setup, working with the angle $\theta' \coloneqq \arcsin (a/c)$ that satisfies $\theta' \leq 0.001$.}
\vspace{0.2cm}

Set $\epsilon' \coloneqq \epsilon/c$, $J \coloneqq \bigg \lfloor \log_{1.05}(0.9/2\epsilon') \bigg \rfloor + 10 $ and $j_\mathrm{end} \coloneqq -1$\;

\For{$j = 0,1,\dots, J$}{
	Let $(2k_j + 1)$ be the largest odd integer less than or equal to $1.05^j$\;
	Prepare $G^{k_j}\ket{\psi'}$ and measure. Repeat at least $5000\cdot\log(5/\gamma)$ times\;
	\If{at least $95\%$ measurements collapsed onto the $1$-subspace}{Set $j_\mathrm{end} \coloneqq j$\;  Exit the loop\;}
}

\eIf{$j_\mathrm{end} = -1$}{
	\Return{$\tilde{a} \coloneqq c \cdot \sin \epsilon' $}
}
{
Initialize $\theta_\mathrm{min} = 0.9 \cdot 1.05^{-j_\mathrm{end}}$ and $\theta_\mathrm{max} = 1.65 \cdot \theta_\mathrm{min}$\;
Set $\epsilon'' \coloneqq \epsilon'/\theta_\mathrm{max}$\;

\For{$j = 0,1,\dots$}{
	Let $q$ be the closest integer to $\theta_\mathrm{min}/2(\theta_\mathrm{max}-\theta_\mathrm{min})$\;
	Let $(2k_j + 1)$ be the closest odd integer to $q\pi/\theta_\mathrm{min}$\;
	Prepare $G^{k_j}\ket{\psi'}$ and measure. Repeat at least $250\cdot\log(1/\gamma_j)$ times, where $\gamma_j \coloneqq 0.9^{-j} \cdot (2\gamma\epsilon''/15$)\;
	Let $\eta_j \coloneqq \theta_\mathrm{max}/\theta_\mathrm{min} - 1$\;
	\eIf{at least $12\%$ measurements collapsed onto the $1$-subspace}{Set $\theta_\mathrm{min} \coloneqq \theta_\mathrm{max}/(1+0.9\eta_j)$\;}{Set $\theta_\mathrm{max} \coloneqq (1+0.9\eta_j)\theta_\mathrm{min}$\; }
	\If{$\theta_\mathrm{max} \leq (1+2\epsilon'')\theta_\mathrm{min}$}{Exit the loop\;}
}

Set $\tilde{\theta}' \coloneqq (\theta_\mathrm{max}+\theta_\mathrm{min})/2$\;

\Return{$\tilde{a} \coloneqq c \cdot \sin \tilde{\theta}'$}
}
\caption{Quantum amplitude estimation (adapted from Ref.~\cite{AR_2020_qae})}
\label{algo:amplitude_estimation}
\end{algorithm}

\section{Adapted quantum exact counting algorithm \label{app:exact_count}}

The error on $t$ that follows from the formulation of QAE by~\textcite{Brassard_2002} is symmetric under the transformation $t \to N-t$, i.e., it does not matter if we are counting the number of $1$- or $0$-inputs. Their counting algorithm makes use of that feature. In contrast, the alternative formulation of QAE by~\textcite{AR_2020_qae} (see Algorithm~\ref{algo:amplitude_estimation}) produces an estimate $\tilde{t} \coloneqq N \cdot \tilde{a}^2$ with error satisfying
\begin{equation}
	\label{eq:count_bound}
	|\tilde{t} - t| \leq 2\epsilon \sqrt{t\left(N - \tfrac{t}{c^2}\right)} + N \epsilon^2,
\end{equation}
with $c = 1001$. The symmetry does not hold. In fact, if we estimate the amplitude $a_0 = \sqrt{1-t/N}$ associated with the $0$-inputs instead, the estimate $\tilde{t}_0 \coloneqq N \cdot (1-\tilde{a}_0^2)$ for $t$ satisfies
\begin{equation}
	\label{eq:count_bound2}
	|\tilde{t}_0 - t| \leq 2\epsilon \sqrt{(N-t)\left(N - \tfrac{N-t}{c^2}\right)} + N \epsilon^2,
\end{equation}
therefore, leading to a smaller error bound if $t > N/2$.

Considering that disparity, we have to adapt the counting algorithm for the alternative formulation of QAE.
We decide whether the count uses the number of $1$-inputs or that of $0$-inputs based on two initial rough estimates for these numbers. The two rough estimates are produced with two independent runs of QAE with error set as $\epsilon' = (2\sqrt{N})^{-1}$ and failure probability bounded by $\gamma/4$.  Let $\tilde{a}'_1$ be the estimate for $a$, and $\tilde{a}'_2$ that for $a_0$.  In case of successful runs, the resulting estimates satisfy
\begin{equation}
	\label{eq:rough_estimate_1}
	|\sqrt{N}\cdot\tilde{a}'_1 - \sqrt{t}| \leq \tfrac{1}{2} \text{ and } |\sqrt{N}\cdot\tilde{a}'_2 - \sqrt{N-t}| \leq \tfrac{1}{2}.
\end{equation}

If $\tilde{a}'_1 < \tilde{a}'_2$, we count the number of $1$-inputs directly. We employ QAE using an error $\epsilon = (8\sqrt{N}(\sqrt{N}\cdot\tilde{a}'_1+\tfrac{1}{2}))^{-1}$ and failure probability $\gamma/2$. By combining inequalities from Eq.~\eqref{eq:count_bound} and Eq.~\eqref{eq:rough_estimate_1}, it follows that the estimate $\tilde{t}$ satisfies
$|\tilde{t} - t| < \tfrac{1}{2}$.

Otherwise, if $\tilde{a}'_1 \geq \tilde{a}'_2$, we use the number of $0$-inputs to produce the count. This time, by setting the error as $\epsilon = (8\sqrt{N}(\sqrt{N}\cdot\tilde{a}'_2+\tfrac{1}{2}))^{-1}$, we produce an estimate $\tilde{t}_0$ that also satisfies $|\tilde{t}_0 - t| < \tfrac{1}{2}$.

In either case, when the QAE runs are successful, the closest integer to the final estimate is exactly $t$. The probability of failing at least one of the QAE runs is smaller than $\gamma/4 + \gamma/4 + \gamma/2 = \gamma$, which follows from the union bound.

In terms of complexity, the first two QAE runs use $\mathcal{O}(\sqrt{N}\log(1/\gamma))$ quantum queries to~$f$. The cost of the final run depends on which rough estimate was smaller; it can be written as $\mathcal{O}(\sqrt{ N \min\{t+1,N-t+1\} }\log(1/\gamma))$, which is in
\begin{equation}
	\mathcal{O}\left(\sqrt{(t+1)(N-t+1)}\log(1/\gamma)\right).
\end{equation}
The final QAE run dominates the total query complexity of this algorithm, which leads to the result stated in Theorem~\ref{cor:quantum_counting}. Algorithm~\ref{algo:exact_counting} summarizes this procedure.
\begin{algorithm}
\SetAlgoLined
\KwIn{$\gamma \in ]0,1]$ and access to the oracle $O_f$ that marks the $1$-inputs of a Boolean function $f$, with domain $\mathcal{T} = \{1,\dots,N\}$.}
\KwOut{The number $t$ of $1$-inputs of $f$.}
\BlankLine
Set $\epsilon' \coloneqq 1/(2\sqrt{N})$\;

\vspace{0.2cm}
\tcp{Rough estimate for the amplitude associated with the $1$-inputs}
\vspace{0.2cm}

Let $\tilde{a}'_1$ be the estimate for $\sqrt{t/N}$ produced by QAE with error $\epsilon'$ and failure probability $\gamma/4$\;

\vspace{0.2cm}
\tcp{Rough estimate for the amplitude associated with the $0$-inputs}
\vspace{0.2cm}

Let $\tilde{a}'_2$ be the estimate for $\sqrt{1-t/N}$ produced by QAE with error $\epsilon'$ and failure probability $\gamma/4$\;

\vspace{0.2cm}

\eIf{$\tilde{a}'_1 < \tilde{a}'_2$}{
	
\vspace{0.2cm}
\tcp{Count using the number of $1$-inputs}
\vspace{0.2cm}

Set $\epsilon \coloneqq 1/(8\sqrt{N}(\sqrt{N}\cdot \tilde{a}'_1 + \tfrac{1}{2}))$\;
Let $\tilde{a}$ be the estimate for $\sqrt{t/N}$ produced by QAE with error $\epsilon$ and failure probability $\gamma/2$ \;
Set $\tilde{t} \coloneqq N \cdot \tilde{a}^2$\;
}{\If{$\tilde{a}'_1 \geq \tilde{a}'_2$}{
	
\vspace{0.2cm}
\tcp{Count using the number of $0$-inputs}
\vspace{0.2cm}

Set $\epsilon \coloneqq 1/(8\sqrt{N}(\sqrt{N}\cdot \tilde{a}'_2 + \tfrac{1}{2}))$\;
Let $\tilde{a}$ be the estimate for $\sqrt{1-t/N}$ produced by QAE with error $\epsilon$ and failure probability $\gamma/2$ \;
Set $\tilde{t} \coloneqq N \cdot(1 - \tilde{a}^2)$\;
}}
Set $\hat{t}$ as the closest integer to $\tilde{t}$\;
\Return{$\hat{t}$.}
\caption{Quantum exact counting (adapted from Ref.~\cite{Brassard_2002})}
\label{algo:exact_counting}
\end{algorithm}

\section{Pseudocode for density peak clustering algorithm \label{app:clustering}}

\begin{algorithm}[H]
\SetAlgoLined
\KwIn{a dataset $S$, a distance function $d$, threshold parameters $\rho_c,\delta_c$ (and the scale parameter $d_c$ depending on the kernel)}
\KwOut{a partition of $S$ into clusters.}
\BlankLine

For all $x_i \in S$ assign a density $\rho_i$\;
For all $x_i \in S$, find the nearest-higher $h_i$ and query the respective nearest-higher separation $\delta_i$\;
For all $x_i \in S$, promote to root if $\delta_i \geq \delta_c$ and $\rho_i \geq \rho_c$, and demote to outlier if $\delta_i \geq \delta_c$ and $\rho_i < \rho_c$\;
For each root create a cluster\;
Assign $x_i \in S$ that is neither a root nor an outlier to the same cluster as $h_i$\;

\Return{clusters.}

\caption{Density peak clustering (original proposed in Ref.~\cite{Laio2014})}
\label{algo:dpc}
\end{algorithm}

\section{Proof of Claim \ref{claim_density}}
\label{sec:appA}

\begin{proof}[Proof of Claim~\ref{claim_density}]
Consider the circuit shown in Fig.~\ref{fig:BE_circuit}, and the initial state $\ket{0}\ket{0^{n_\chi}}\ket{0^{n_d}}\ket{0^{n}}\ket{i}$. The first Hadamard gates create a uniform superposition of all data points $x_j \in S$ in the fourth register, leading to
    \begin{equation}
        \frac{1}{\sqrt{N}} \sum_{x_j \in S} \ket{0}\ket{0^{n_\chi}}\ket{0^{n_d}}\ket{j}\ket{i}.
    \end{equation}
    The distance oracle $O_d$ is then applied for a first time to store all distances from $x_i$ in the third register,
        \begin{equation}
        \frac{1}{\sqrt{N}} \sum_{x_j \in S} \ket{0}\ket{0^{n_\chi}}\ket{d(x_i,x_j)}\ket{j}\ket{i}
        \end{equation}
    The application of $U_\chi$ follows, storing the kernel values in the second register, which controls the rotation $R_\chi$ applied to the first register, resulting in
        \begin{equation}
        \frac{1}{\sqrt{N}} \sum_{x_j \in S} \left(\chi_{ij}\ket{0} + \sqrt{1-\chi_{ij}^2}\ket{1}\right)\ket{\chi_{ij}}\ket{d(x_i,x_j)}\ket{j}\ket{i}.
        \end{equation}
    The following step consists of uncomputing the second and the third register, eliminating their entanglement with the remaining state. This is done by reapplying $U_\chi$ and $O_d$ on the same registers, since $y \oplus y = 0$ for any bit string $y$. Then, by applying the Hadamard gates again we have
    \begin{equation}
        \frac{1}{N} \sum_{x_k \in S} \sum_{x_j \in S} (-1)^{j \cdot k}
		 \left(\chi_{ij}\ket{0} + \sqrt{1-\chi_{ij}^2}\ket{1}\right) \ket{0^{n_\chi}}\ket{0^{n_d}} \ket{k}\ket{i},
    \end{equation}
    where $j \cdot k$ denotes the bitwise inner product between the two integers $j$ and $k$. The only component having the first $n_\rho$ qubits in the zero-state is
    \begin{equation}
        \frac{1}{N}\sum_{x_j \in S} \chi_{ij} \ket{0^{n_\rho}} \ket{i} =  \rho_i \ket{0^{n_\rho}} \ket{i},
    \end{equation}
    where the equality comes from the definition of density in Eq.~\eqref{eq:density}. This allows us to write the resulting state compactly as
    \begin{equation}
        \left(\rho_i \ket{0^{n_\rho}}  + \sqrt{1-\rho_i^2}\ket{\perp}\right) \ket{i}
    \end{equation}
    with $\ket{\perp}$ being such that $\braket{0^{n_\rho}}{\perp} = 0$.

In this procedure we used $\mathcal{O}(\log N)$ Hadamard gates to create the equal superposition states, and we made two queries to $O_d$.
\end{proof}

\section{Pseudocode for decision clustering algorithm\label{app:decisionclustering}}
\begin{algorithm}[H]
\SetAlgoLined
\KwIn{a dataset $S$, a distance function $d$, two points $x_i,x_j \in S$, threshold parameters $\rho_c,\delta_c$ (and the scale parameter $d_c$ depending on the kernel)}
\KwOut{$\mathrm{Yes}/\mathrm{No}$.}
\BlankLine

\vspace{0.2cm}
\tcp{Let $FindRoot$ be the function that, given $x_i \in S$, builds the path of nearest-highers until reaching a root, and returns it. If the point is already a root or an outlier, the function returns it immediately.  }
\vspace{0.2cm}

\Fn{FindRoot($x_i$)}{
	\eIf{$x_i$ is a root or $x_i$ is an outlier}{
		\Return $x_i$\;
	}{
		\Return $FindRoot(h_i)$\;
	}
}
\vspace{0.2cm}

\eIf{$FindRoot(x_i) = FindRoot(x_j)$}{
	\Return{$\mathrm{Yes}$}\;}{\Return{$\mathrm{No}$}.}

\caption{Decision clustering (original proposed in Ref.~\cite{Magano_2023})}
\label{algo:decision_clustering}
\end{algorithm}

\section{Pseudocode for quantum cluster counting algorithm \label{app:clustercounting}}
\begin{algorithm}[H]
\SetAlgoLined
\KwIn{a dataset $S$, a distance function $d$, two points $x_i,x_j \in S$, threshold parameters $\rho_c,\delta_c$ (and the scale parameter $d_c$ depending on the kernel)}
\KwOut{The number $K$ of clusters.}
\BlankLine

Apply the exact counting algorithm (Theorem~\ref{cor:quantum_counting}) with the oracle $O_{f_r}$, and allowing for a failure probability of $\gamma$. Let $\widehat{K}$ be the output.

\Return{$\widehat{K}$.}
\caption{Quantum cluster counting}
\label{algo:clustercounting}
\end{algorithm}

\section{Root identification oracle \label{sec:root_oracle}}

Let $U_r$ be the unitary that evaluates $f_r$ coherently, acting as
\begin{equation}
    U_r\ket{b}\ket{i} = \ket{b \oplus f_r(x_i)}\ket{i},
\end{equation} 
with $b=0,1$. The root identification oracle $O_{f_r}$ can be implemented with one application of $U_r$ as represented in the following circuit (this oracle construction is commonly called phase kickback).
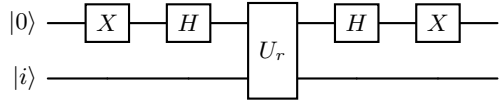
\begin{figure}[!htbp]
		\centering
		\scalebox{1}{ 
		\begin{quantikz}[row sep=0.2cm, column sep=0.5cm]
				\lstick{$\ket{0}$} & \gate{X}  & \gate{H} & \gate[wires = 2]{U_r}  & \gate{H} & \gate{X} &  \\
				\lstick{$\ket{i}$} & & &  & & & \\
		\end{quantikz}
		}
		\caption{Circuit implementing the oracle $O_{f_r}$.}
		\label{fig:root_oracle}
\end{figure}
We now show how we can implement a unitary $\widetilde{U}_r$ that approximates $U_r$ with bounded error. First, suppose that all densities have already been computed beforehand and are stored in memory; let $O_\rho$ be the unitary that registers the density $\rho_i$ controlled by a second register carrying the label $i$, i.e.,
\begin{equation}
	O_\rho \ket{0^{n_\epsilon}}\ket{i} = \ket{\rho_i}\ket{i},
\end{equation} 
where $n_\epsilon$ is the number of bits required to register a density with accuracy $\epsilon$. A call to $O_\rho$ does not require any additional distance query. 

	The nearest-higher $h_i$ can also be registered coherently. The nearest-higher search uses the quantum minimum finding~\cite{durr1999}, which admits a coherent formulation.
	Let $U_\textrm{NH}$ be the unitary realized by the quantum circuit that reproduces the quantum algorithm for finding the minimum of $f_i$ (i.e., the nearest-higher $h_i$) controlled by a second register carrying the label $i$. Its action can be written as
\begin{equation}
    U_\textrm{NH}\ket{0^n}\ket{i} = \left(\sqrt{p}\ket{h_i} + \sqrt{1-p}\ket{\perp}\right)\ket{i},
\end{equation}
with $p \geq 1-\gamma$, the probability of finding the correct nearest-higher, and with $\ket{\perp}$ orthogonal to $\ket{h_i}$. The implementation of $U_\textrm{NH}$ has the same query complexity as the quantum nearest-higher search, which is $\mathcal{O}(\sqrt{N}\log(1/\gamma))$.   

By composing $O_\rho$ and $U_\textrm{NH}$ together with the distance oracle $O_d$, we implement the unitary $\widetilde{U}_r$. The cost of this implementation and the error in the approximation are established in the following claim (note that the cost of pre-evaluating all densities is not being considered here).
\begin{claim}
    \label{claim:unitary_fr}
        Suppose that densities have been assigned to all dataset elements, and that their values can be accessed coherently. Given $\Delta > 0$, it is possible to implement a unitary $\widetilde{U}_r$ satisfying
    \begin{equation*}
            \bigg|\bigg|\widetilde{U}_r - U_r\bigg|\bigg|_2 \leq \Delta,
    \end{equation*}~with $\mathcal{O}(\sqrt{N}\log(1/\Delta))$ distance queries.
\end{claim}

\begin{proof}
    Consider the initial state $\ket{b}\ket{0^{n_\epsilon}}\ket{0^{n_d}}\ket{0^n}\ket{i}$. We first apply the unitary $U_\textrm{NH}$ on the last two registers, which leads to
    \begin{equation}
        \ket{b}\ket{0^{n_\epsilon}}\ket{0^{n_d}}\left(\sqrt{p} \ket{h_i} + \sqrt{1-p}\ket{\perp}\right) \ket{i}.
    \end{equation}
    Then, we act with $O_d$ on the last three registers,
    \begin{equation}
        \sqrt{p} \ket{b}\ket{0^{n_\epsilon}}\ket{\delta_i}\ket{h_i}\ket{i} + \sqrt{1-p}\ket{\perp^{(1)}}\ket{i},
    \end{equation}
	where we used the definition of nearest-higher separation given in Eq.~\eqref{eq:nearest_higher_sep}. The state $\ket{\perp^{(1)}}\ket{i}$ remains normalized and orthogonal to the first component. We register the density on the second register using the unitary $O_\rho$,
    \begin{equation}
        \sqrt{p} \ket{b}\ket{\rho_i}\ket{\delta_i}\ket{h_i}\ket{i} + \sqrt{1-p}\ket{\perp^{(2)}}\ket{i},
    \end{equation}
    where $\ket{\perp^{(2)}}\ket{i}$ is again normalized and orthogonal to the first component. Controlled on the third and fourth registers, we can write $f_r(x_i)$ on the first register, resulting in
    \begin{equation}
        \sqrt{p}\ket{b\oplus f_r(x_i)}\ket{\rho_i}\ket{\delta_i}\ket{h_i}\ket{i} + \sqrt{1-p}\ket{\perp^{(3)}}\ket{i},
    \end{equation}
     where $\ket{\perp^{(3)}}\ket{i}$ is again normalized and orthogonal to the first component. We uncompute the second, the third and the fourth registers by acting with $O_\rho$, followed by $O_d$, and then by $U_\mathrm{NH}^\dagger$. The resulting state can be written compactly with the omission of the ancilla registers,
    \begin{equation}
        p\ket{b \oplus f_r(x_i)}\ket{i} + \sqrt{1-p^2}\ket{\perp^{(4)}}\ket{i},
    \end{equation}
    and again, $\ket{\perp^{(4)}}\ket{i}$ is normalized and orthogonal to the first component.

	Let $\widetilde{U}_r$ be the unitary resulting from the sequence of operations we have just described. It follows
    \begin{equation}
        \begin{split}
        \bigg|\bigg|  (\widetilde{U}_r  - U_r)\ket{b}\ket{i}  \bigg|\bigg|^2_2 = 2-2p \leq 2\gamma,
        \end{split} 
    \end{equation}
    where we used $p \geq 1-\gamma$. Since the control register is preserved, the same bound applies when we consider its action on any normalized superposition $\sum_{x_i \in S} \alpha_i \ket{0}\ket{i}$,
\begin{equation}
	\begin{split}
    \bigg|\bigg| (\widetilde{U}_r - U_r)\sum_{x_i \in S} \alpha_i \ket{b}\ket{i} \bigg| \bigg|^2_2 &\leq \sum_{x_i \in S} |\alpha_i|^2 \bigg|\bigg|(\widetilde{U}_r - U_r)\ket{b}\ket{i} \bigg|\bigg|^2_2 \\ &\leq 2\gamma.
	\end{split}
\end{equation}
    The query complexity comes from the two applications of $U_\mathrm{NH}$ and of $O_d$; this represents a total of $\mathcal{O}(\sqrt{N}\log(1/\gamma))$ distance queries.
 The inequality in the claim's statement is obtained by setting $\gamma = \Delta^2/2$.
\end{proof}

The oracle $O_{f_r}$ can be written as a sequence $U_1 U_r U_2$, with $U_1 =  (XH \otimes I_n)$ and $U_2 = (HX \otimes I_n)$. Let $\widetilde{O}_{f_r}$ be the unitary defined by the sequence $U_1 \widetilde{U}_r U_2$. It follows that
\begin{equation}
	\bigg|\bigg|\widetilde{O}_{f_r} - O_{f_r}\bigg|\bigg|_2 \leq \Delta,
\end{equation}
which concludes the proof of Claim~\ref{claim:root_oracle}.

\section{Proof of Claim~\ref{claim:6}}
		\label{sec:appC}
		Before proving this claim, we first prove the following auxiliary result.

		\begin{proposition}
			\label{proposition:8}
			Let $a,b > 0$, and consider the inequality \[ t \geq a \log t + b .\] Define \[\tau \coloneqq b + 2a\log(a+b+4).\] The inequality is satisfied for any $t \geq \tau$.
		\end{proposition}
		\begin{proof}
			Let $W \coloneqq a + b + 4 \geq 4$, such that $\tau = b+2a\log W$, and define $f(t) \coloneqq  t - a \log t - b$. This function is increasing for $t \geq a$. In particular, it is increasing for $t \geq \tau \geq a$. It then suffices to prove that $f(\tau) \geq 0$. By evaluating $f$ at $\tau$ one gets $f(\tau) = a \log (W^2/\tau)$. In turn, it is sufficient to verify that $W^2 \geq \tau$. Given that $\tau \leq W + 2W\log W$, this is verified for $W$ such that $W \geq 1+ 2\log W$, which can be seen to be true for every $W \geq 4$.
		\end{proof}
		We now move on to the proof of Claim~\ref{claim:6}. The proof makes use of the following result from Ref.~\cite{McAllesterSchapire2000}. 
		\begin{theorem}[Theorem 9 in Ref.~\cite{McAllesterSchapire2000}]
			\label{thm:1}
			Given $\gamma > 0$ and $t$ samples,
			\[\mathrm{P}\left[M_0(t) > G_0(t) + (2\sqrt{2} + \sqrt{3}) \sqrt{\frac{1}{t}\log \frac{3}{\gamma}}\right] \leq \gamma \ . \]
		\end{theorem}
		Relying on this theorem, the proof goes as follows.
		\begin{proof}[Proof of Claim~\ref{claim:6}]
		Define
		\begin{equation}
		m_t \coloneqq (2\sqrt{2} + \sqrt{3}) \sqrt{\frac{1}{t}\log \frac{12t^2}{\gamma}},
		\end{equation}
		such that it follows from Theorem~\ref{thm:1} that
		\begin{equation}
			\mathrm{P}[M_0(t) > G_0(t) + m_t] \leq \frac{\gamma}{4t^2},
		\end{equation}
		for $t > 0$ a specific number of samples.

		Let $c \coloneqq (4\sqrt{2} + 2\sqrt{3})^2$, and $\tau$ be given by
		\begin{equation}
		\tau \coloneqq \frac{c}{m^2} \log \left(\frac{12}{\gamma}\right) + \frac{4c}{m^2} \log \left(\frac{2c}{m^2} + \frac{c}{m^2}\log \left(\frac{12}{\gamma}\right) + 4\right),
		\end{equation}
		which is upper bounded by 
		\begin{equation}
		\tau_1 \coloneqq \frac{5400}{m^2} \log\left(\frac{2}{m\gamma}\right).
		\end{equation}
		From Proposition~\ref{proposition:8} it follows that for any $t \geq \tau_1 \geq \tau$,
		\begin{equation}
			m_t \leq m/2,
		\end{equation}
		so that it follows
		\begin{equation}
			P[M_0(t) > G_0(t) + m/2] \leq P[M_0(t) > G_0(t) + m_t] \leq \frac{\gamma}{4t^2}.
		\end{equation}

		\end{proof}

\section{Proof of Claim~\ref{claim:5}}
		\label{sec:appB}

		The proof of this claim relies on the same type of machinery that was used to prove the previous one. Instead of using the result from Ref.~\cite{McAllesterSchapire2000} as the starting point, we make use of Hoeffding's inequality~\cite{hoeffding1963}.

		\begin{proof}[Proof of Claim~\ref{claim:5}]
		Define
		\begin{equation}
			\eta_t \coloneqq \sqrt{\frac{1}{2t} \log \frac{8Nt^2}{\gamma}},
		\end{equation}
		such that it follows from Hoeffding's inequality that
		\begin{equation}
			\mathrm{P}\left[|\hat{m}_k(t) - m_k| \geq \eta_t \right] \leq \frac{\gamma}{4Nt^2}, \text{ with } k = 1,\dots, K
		\end{equation}
		and where $t>0$ is a specific number of samples. 

		Let $\tau$ be given by
		\begin{equation}
			\tau \coloneqq \frac{1}{2\eta^2} \log \frac{8N}{\gamma} + \frac{2}{\eta^2} \log\left(\frac{1}{\eta^2} + \frac{1}{2\eta^2} \log \frac{8N}{\gamma} + 4\right),
		\end{equation}
		which is upper bounded by
		\begin{equation}
			\tau_2 \coloneqq \frac{16}{\eta^2} \log \left(\frac{N}{\eta \gamma}\right).
		\end{equation}
		From Proposition~\ref{proposition:8} it follows that for any $t \geq \tau_2 \geq \tau$,
		\begin{equation}
			\eta_t \leq \eta.
		\end{equation}
		so that it follows
		\begin{equation}
			\mathrm{P}\left[|\hat{m}_k(t) - m_k| \geq \eta \right] \leq \frac{\gamma}{4Nt^2}
		\end{equation}

		The probability of having some estimator escaping the error bound can be upper-bounded according to the union bound, i.e.,

		\begin{equation}
			\begin{split}
				\mathrm{P}\left[\bigcup_{k=1}^K\left\{|\hat{m}_k(t) - m_k| \geq \eta\right\}\right] &\leq \sum_{k=1}^K\mathrm{P}\left[|\hat{m}_k(t) - m_k| \geq \eta\right]\\
				&\leq  \sum_{k=1}^K \frac{\gamma}{4Nt^2} \leq \frac{\gamma}{4t^2}.
			\end{split}
		\end{equation}

		\end{proof}

\section{Pseudocode for heavy cluster counting algorithm \label{app:heavycounting}}

\begin{algorithm}[H]
\SetAlgoLined
\KwIn{dataset S, distance function $d$, and $m,\eta,\gamma \in ]0,1]$.}
\KwOut{An estimate $\widehat{K_m}$ for the $m$-heavy cluster count.}
\BlankLine

\vspace{0.2cm}
\tcp{Let $FindRoot$ be the function that, given $x_i \in S$, builds the path of nearest-highers until reaching a root, and returns it. If the point is already a root or an outlier, the function returns it immediately.  }
\vspace{0.2cm}

\Fn{FindRoot($x_i$)}{
	\eIf{$x_i$ is a root or $x_i$ is an outlier}{
		\Return $x_i$\;
	}{
		\Return $FindRoot(h_i)$\;
	}
}

\vspace{0.2cm}
\tcp{Let $\tau_1$ and $\tau_2$ be defined according to Claim~\ref{claim:6} and Claim~\ref{claim:5}, respectively.}
\vspace{0.2cm}

Set  $\tau \coloneqq \max\{\tau_1,\tau_2\}$, $t \coloneqq 0$, $\phi_1 \coloneqq 0$ and $G_0 \coloneqq 1 $ \; 
\vspace{0.2cm}
\While{$t < \tau$ or $G_0 > m/2$}{
	Draw $x_i$ uniformly at random from $S$ and update $t \coloneqq t+1$\;
	Set $y_k \coloneqq FindRoot(x_i)$\;
	
	\vspace{0.2cm}
	\tcp{If $y_k$ is a root, it identifies the cluster $C_k$, with $k = 1,\dots, K$. Otherwise, it is an outlier and identifies $C_0$.}
	\vspace{0.2cm}

	\eIf{$C_k$ is seen for the first time}{Update $\phi_1 \coloneqq \phi_1 + 1$ \;}{\If{$C_k$ is seen for the second time}{Update $\phi_1 \coloneqq \phi_1 - 1$\;}}
	Update $G_0 \coloneqq \phi_1/t$\;
}
\vspace{0.2cm}
For every cluster $C_k \neq C_0$ estimate $\hat{m}_k$\;
\Return{$\widehat{K_m} \coloneqq |\{k=1,\dots, K : \hat{m}_k > m\}|$.}
\caption{Heavy cluster counting}
\label{algo:heavycounting}
\end{algorithm}

	\section{Proof of Claim~\ref{claim:count}}
		\label{sec:appD}
		\begin{proof}[Proof of Claim~\ref{claim:count}]

		Algorithm~\ref{algo:heavycounting} takes at least $\tau_\mathrm{min} \coloneqq \max \{\tau_1,\tau_2\}$ samples. Let $T \in \{\tau_\mathrm{min}, \tau_\mathrm{min}+1, \dots\}$ be the random number of samples taken by the algorithm. We define $M_0(T)$, $G_0(T)$ and $\hat{m}_k(T)$ as the missing mass, the Good-Turing estimator, and the mass estimate of the $k$-th cluster at the end of the algorithm, respectively.
		
		The algorithm returns a wrong solution if either $M_0(T)>m$ or some estimate $\hat{m}_k(T)$ has an error of at least $\eta$. By the union bound, the probability of failure is then upper-bounded as
		\begin{equation}
			\mathrm{P}[\text{algorithm fails}] \leq \mathrm{P}[M_0(T)>m] + \mathrm{P}\left[\bigcup_{k=1}^K\left\{|\hat{m}_k(T) - m_k| \geq \eta\right\}\right].
		\end{equation}
		By construction, the algorithm halts only when $G_0(T) \leq m/2$, and therefore
		\begin{equation}
			\label{eq:wrong_sol}
			\mathrm{P}[\text{algorithm fails}] \leq \mathrm{P}[M_0(T)> G_0(T) + m/2] + \mathrm{P}\left[\bigcup_{k=1}^K\left\{|\hat{m}_k(T) - m_k| \geq \eta\right\}\right].
		\end{equation}

		Let $(X_1,X_2,\dots)$ be an infinite random sample of clusters drawn independently from the dataset according to their mass distribution. We denote by $(x_1^{(t)},x_2^{(t)},\dots)$ any realization of that infinite sample that makes the algorithm halt immediately after taking $t \geq \tau_\mathrm{min}$ samples. We define the random variable $A_T$ that evaluates to $1$ when the condition $M_0(T)> G_0(T) + m/2$ is satisfied, and evaluates to $0$ otherwise. Analogously, for a fixed number of samples $t$, $A_t$ evaluates to $1$ only when the condition $M_0(t)> G_0(t) + m/2$ is satisfied, and to $0$ otherwise.

		The first term on the right-hand side of Eq.~\eqref{eq:wrong_sol} can be expressed in terms of these variables as
		\begin{equation}
			\begin{split}
			\mathrm{P}[M_0(T)> G_0(T) + m/2] = \mathrm{P}[A_T = 1] &= \sum_{t = \tau_\mathrm{min}}^\infty \, \sum_{x_1^{(t)},   x_2^{(t)},   \dots} \, \mathrm{P}\left[X_1=x_1^{(t)}, X_2=x_2^{(t)},\dots\right] \, A_T\left(x_1^{(t)},  x_2^{(t)},\dots\right) \\
			&= \sum_{t = \tau_\mathrm{min}}^\infty \, \sum_{x_1^{(t)}, x_2^{(t)}, \dots} \, \mathrm{P}\left[X_1=x_1^{(t)},X_2=x_2^{(t)},\dots\right] \, A_t \left(x_1^{(t)},x_2^{(t)},\dots \right). \\
			\end{split}
		\end{equation}
		Whether the algorithm halts after $t$ samples depends exclusively on those $t$ samples, and, likewise, the value of $A_t$ depends only on the first $t$ samples. Therefore, the expression above can be rewritten as
		\begin{equation}
			\sum_{t = \tau_\mathrm{min}}^\infty \, \sum_{x_1^{(t)}, \dots, x_t^{(t)}} \, \mathrm{P}\left[X_1=x_1^{(t)}, \dots, X_t=x_t^{(t)}\right] \, A_t\left(x_1^{(t)},\dots x_t^{(t)}\right),
		\end{equation}
		where we used the fact that the samples are independent and that the probabilities of all possible sequences of samples following the $t$-th one sum to $1$. Clearly, the set of sequences of $t$ samples that make the algorithm halt is a subset of all possible sequences of $t$ samples. Extending the summation to all possible sequences of $t$ samples therefore yields an upper bound on the expression above,
		\begin{equation}
			\sum_{t = \tau_\mathrm{min}}^\infty \, \sum_{x_1^{(t)}, \dots, x_t^{(t)}} \, \mathrm{P}\left[X_1=x_1^{(t)}, \dots, X_t=x_t^{(t)}\right] \, A_t\left(x_1^{(t)},\dots, x_t^{(t)}\right) \leq \sum_{t = \tau_\mathrm{min}}^\infty \, \sum_{x_1, \dots, x_t} \, \mathrm{P}\left[X_1=x_1, \dots, X_t=x_t\right] \, A_t\left(x_1,\dots, x_t\right).
		\end{equation}
		This upper bound can be written more compactly as
		\begin{equation}
			\begin{split}
			\sum_{t = \tau_\mathrm{min}}^\infty \, \sum_{x_1, \dots, x_t} \, \mathrm{P}\left[X_1=x_1, \dots, X_t=x_t\right] \, A_t\left(x_1,\dots, x_t\right) &= \sum_{t = \tau_\mathrm{min}}^\infty \, \mathrm{P}\left[A_t=1\right] \\ &=  \sum_{t = \tau_\mathrm{min}}^\infty \, \mathrm{P}\left[M_0(t)> G_0(t) + m/2\right].
			\end{split}
		\end{equation}
		By the definition of $\tau_\mathrm{min}$, every term of the sum satisfies $t \geq \tau_1$. Hence, Claim~\ref{claim:6} yields the upper bound
		\begin{equation}
			\begin{split}
				\sum_{t = \tau_\mathrm{min}}^\infty \, \mathrm{P}\left[M_0(t)> G_0(t) + m/2\right] \leq \sum_{t = \tau_\mathrm{min}}^\infty \frac{\gamma}{4t^2} \leq \frac{\gamma}{2}.
			\end{split}
		\end{equation}
		Thus, the first term in Eq.~\eqref{eq:wrong_sol} is upper-bounded by $\gamma/2$. An identical argument establishes the same upper bound on the second term
		\begin{equation}
			\begin{split}
			\mathrm{P}\left[\bigcup_{k=1}^K\left\{|\hat{m}_k(T) - m_k| \geq \eta\right\}\right] \leq \sum_{t = \tau_\mathrm{min}}^\infty \mathrm{P}\left[\bigcup_{k=1}^K\left\{|\hat{m}_k(t) - m_k| \geq \eta\right\}\right] 
			\leq \sum_{t = \tau_\mathrm{min}}^\infty  \frac{\gamma}{4t^2} \leq \frac{\gamma}{2}
			\end{split}
		\end{equation}
		where the second inequality follows from Claim~\ref{claim:5}, which is holds because $t\geq \tau_2$ in each term of the sum. For the failure probability of the algorithm it follows
		\begin{equation}
			\mathrm{P}[\text{algorithm fails}] \leq \frac{\gamma}{2} + \frac{\gamma}{2} = \gamma.
		\end{equation}

		We now analyze the sample complexity of the algorithm. Recall that, after $\tau_\mathrm{min}$ samples, sampling continues until $G_0(T)\leq m/2$. By definition,
		\begin{equation}
			G_0(T) \leq \frac{K+1}{T}.
		\end{equation}
		Defining $\tau_3 \coloneqq 2(K+1)/m$, it follows that $G_0(T)\leq m/2$ for all $T \geq \tau_3$. Therefore, the number of samples taken by the algorithm satisfies
		\begin{equation}
			\tau_\mathrm{min} \leq T \leq \max\{\tau_\mathrm{min}, \tau_3 \} = \mathcal{O}\left(\frac{K}{m} + \frac{1}{\eta^{2}} \log\frac{N}{\gamma\eta} \right).
		\end{equation}
			
	\end{proof}

\end{document}